\documentclass[12pt, draftclsnofoot, onecolumn]{IEEEtran}
 \usepackage{amsmath,amssymb}
 \usepackage{graphicx,graphics,color,psfrag}
 \usepackage{cite,balance}
 \usepackage{caption}
 \usepackage{subcaption}
 \allowdisplaybreaks
 \usepackage{algorithm}
 \usepackage{accents}
 \usepackage{amsthm}
 \usepackage{amsmath}
 \usepackage{bm}
 \usepackage{algorithmic}
 \usepackage[english]{babel}
 \usepackage{multirow}
 \usepackage{enumerate}
 \usepackage{cases}
 \usepackage{stfloats}
 \usepackage{dsfont}
 \usepackage{color,soul}
 \usepackage{amsfonts}
 \usepackage{cite,graphicx,amsmath,amssymb}
 \usepackage{fancyhdr}
 \usepackage{hhline}
 \usepackage{graphicx,graphics}
 \usepackage{array,color}
 \usepackage{amsmath}
 \usepackage{booktabs}

\makeatletter
\renewcommand{\fnum@figure}{Fig. \thefigure}
\makeatother

\newtheoremstyle{itshape}
  {.0\baselineskip\@plus.0\baselineskip\@minus.0\baselineskip}
  {.0\baselineskip\@plus.0\baselineskip\@minus.0\baselineskip}
  {\itshape}
  {}
  {\bfseries}
  {.}
  { }
  {}

\theoremstyle{itshape}

\include{header}

\newtheorem{remark}{Remark}
\newtheorem{proposition}{Proposition}
\newtheorem{lemma}{Lemma}

\begin{document}
\title{Multiuser Computation Offloading and Downloading for Edge Computing with Virtualization}

\author{Zezu Liang, Yuan Liu, Tat-Ming Lok, and Kaibin Huang

\thanks{Z. Liang and T. M. Lok are with Department of Information Engineering, The Chinese University of Hong Kong, Hong Kong (e-mail: lz017@ie.cuhk.edu.hk;  tmlok@ie.cuhk.edu.hk). Y. Liu is with school of Electronic and Information Engineering,
South China University of Technology, Guangzhou 510641, China (e-mail: eeyliu@scut.edu.cn).  K. Huang is with Department of Electrical and Electronic Engineering, The University of Hong Kong, Hong Kong (e-mail: huangkb@eee.hku.hk).}
}

\maketitle

\vspace{-1.5cm}
\begin{abstract}
Mobile-edge computing (MEC) is an emerging technology for  enhancing the computational capabilities of the mobile devices and reducing  their energy consumption via offloading complex computation  tasks to the nearby servers. Multiuser MEC at servers is widely realized via parallel computing based on virtualization. Due to finite shared I/O resources, interference between virtual machines (VMs), called I/O interference,  degrades the computation performance. In this paper, we study  the problem of joint radio-and-computation resource allocation (RCRA) in multiuser  MEC systems in the presence of I/O interference. Specifically, offloading scheduling algorithms are designed targeting  two system performance metrics: sum offloading throughput maximization  and sum mobile energy consumption minimization.  Their designs are formulated as non-convex mixed-integer programming problems, which account for latency due to offloading, result downloading and parallel computing. A set of low-complexity algorithms are designed based on a decomposition approach and leveraging classic techniques from combinatorial optimization. The resultant algorithms jointly schedule offloading users, control their offloading sizes, and divide time for communication (offloading and downloading) and computation. They are either optimal or can achieve close-to-optimality as shown  by simulation. Comprehensive simulation results demonstrate considering of I/O interference can endow on an offloading controller robustness against the performance-degradation factor.
\end{abstract}


\section{Introduction}
Driven by the increasing popularity of mobile devices (e.g., smart phones, tablets, wearable devices), a wide range of new mobile applications (e.g., augmented reality, face recognition, interactive online-gaming) are emerging. They usually require intensive computation to enable real-time machine-to-machine and machine-to-human interactions. The limited energy and computation resources at the mobile devices may not be sufficient for meeting the requirement. To address these limitations, mobile-cloud computing (MCC) \cite{survey1} offers one possible  solution by migrating the computation-intensive tasks from mobiles to the cloud. However, data propagation through wide area networks (including the radio-access network, backhaul-network, and Internet)  can cause excessive latency. Therefore, MCC may not be able to support latency-critical applications.

Recently, mobile-edge computing (MEC) \cite{White_Paper,survey}, which provides users computing services using servers at the network edge, is envisioned as a promising way to enable computation-intensive and latency-sensitive mobile applications. Compared with MCC, users in MEC systems offload tasks to the proximate edge servers [e.g., base stations (BSs) and access points (APs)] for execution, which avoids data delivery over the backhaul networks and thereby dramatically reduces latency. An essential technology for implementing MEC is virtualization, referring to sharing of a physical machine (server) by multiple computing processes via the execution of virtual machines (VMs). Specifically, each VM is a virtual computer configured with a certain amount of the server's hardware resource (such as CPU, memory and I/O bus). According to technical standards for the MEC server architecture\cite{White_Paper}, the virtualization functionality is supported by a virtualization layer and a virtualization manager. The virtualization layer virtualizes the MEC hosting infrastructure by abstracting the detailed hardware implementation, while the virtualization manager provides the virtualized computer infrastructure as a service (IaaS). Applications run on top of an IaaS and are deployed within the packaged-operating systems (i.e.,VMs) that are well-isolated with the others. To this end, the MEC server can isolate co-hosted applications and provide multi-service support. Nevertheless, it has been shown in the literature \cite{VM1, netIO2,netIO3} that sharing the same physical platform can incur the so-called I/O interference among VMs, resulting in a certain degree of computation-speed reduction for each VM. As far as we know, prior research of this issue focuses on the interference modeling \cite{netIO2, netIO3,parallel_computing_mode} and computation resource provisioning \cite{8495906}. No previous works related to the computation offloading coupled with joint radio-and-computational resource allocation (RCRA) have been studied before. In this paper, we investigate the multiuser offloading problem in a MEC system in the presence of I/O interference.

\subsection{Prior Work}
In recent years, extensive research has been conducted on efficient computation offloading for MEC systems. For single-user MEC systems, a research focus is designing policies for task assignment or partitioning. A binary-offloading policy (decides on whether an entire task should be offloaded for edge execution or computed locally) has been widely investigated in different system scenarios, including stochastic wireless channels \cite{singleuser1}, MEC systems powered by energy harvesting \cite{Dynamic_JunZhang} or wireless energy transfer \cite{singleuser4}. Based on program partitioning, partial offloading is possible where a computation task at a user can be partitioned into multiple parts for local computing and offloading at the same time. The optimal offloading strategies for partial offloading are studied in \cite{singleuser2,singleuser3}.

For multiuser MEC systems, the efficient computation offloading designs requires joint optimization of RCRA, i.e., how to allocate the finite radio-and-computational resources to multiple users for achieving a system-level objective, e.g., the sum energy consumption minimization. The problem is challenging as multiplicity of parameters and constraints are involved such as multi-user channel states and task information, computation capacities of servers and users, and deadline constraints. In \cite{Multiuser1}, the resource-allocation strategies were proposed based on time-division multiple access (TDMA) and orthogonal frequency-division multiple access (OFDMA). It is assumed that the task-execution durations at the edge cloud are negligible, overlooking the effect of finite computation resources at servers in offloading decisions. In \cite{multiuser3,additional1}, game theory was applied to designing efficient distributed offloading. \cite{multiuser4, multiuser5} studied the multi-cell MEC systems, where joint RCRA under given offloading decisions was optimized in \cite{multiuser4} while \cite{multiuser5} further incorporated offloading decisions into optimization. In \cite{stochastic,additional4}, dynamic offloading policies were proposed to investigate the energy-delay tradeoff for stochastic MEC systems. Energy-efficient offloading designs have also been studied in other scenarios like wireless power transfer \cite{Computation_rate_Bi,wireless_powered} and cooperative transmissions \cite{NOMA,cran}. The work in \cite{multiuser6} is closely related to this paper, as they both address parallel computation at a MEC server for joint RCRA. However, simultaneous computation processes at the same server are assumed in \cite{multiuser6} to be independent and conditioned on partitioned computation resources. The effect of I/O interference is neglected despite its being an importance issue in virtualization.

Omitting I/O interference in multiuser MEC based on virtualization leads to the unrealistic assumption that the total computation resource at a server remains fixed regardless of the number of VMs. In reality, the resource reduces as the number grows due to I/O interference. Thus, the number of VMs per server is usually constrained in practice, so as to maintain the efficiency in resource utilization. Despite its importance, I/O interference has received little attention in the literature. It motivates the current work on accounting for the factor in resource allocation for MEC systems.

\subsection{Contributions and Organization}
In this paper, we revisit the RCRA problem in multiuser offloading and address the following two practical issues that have received scant attention in the literature.

\begin{enumerate}
\item (I/O interference) The I/O interference in practical parallel computing has been largely neglected in the existing ``cake-slicing" model of computing-resource allocation (see e.g., \cite{multiuser6,stochastic}). Considering I/O interference introduces a \emph{dilemma}: scheduling more offloading users increases the multiplexing gain in parallel computing but degrades the speeds of individual VMs due to their interference.
\item (Result downloading) The communication overhead for computation-result downloading is commonly assumed in the literature to be negligible compared with that for offloading. The assumption does not always hold in applications such as augmented reality and image processing. Considering downloading complicates scheduling as the policy needs account for not only users' uplink channel states but also downlink states as well as the output-input-size ratio  for each task.
\end{enumerate}

In this paper, we consider a multiuser MEC system where parallel computing at the sever is based on virtualization. The I/O interference is modelled using a practical model developed based on measurement data \cite{parallel_computing_mode}. While the literature focuses on offloading latency, we consider offloading, parallel computing and downloading as factors contributing to latency. Based on the assumptions, scheduling policies are designed by solving two RCRA problems based on two criteria, namely  \emph{maximizing the sum offloading throughput} and  \emph{minimizing  the sum mobile energy consumption}, both under a latency constraint. The main contributions are summarized as follows.

\begin{itemize}
\item{\bf (Sum Offloading Throughput Maximization)} Based on this criterion, the RCRA problem is formulated as a non-convex problem for joint optimization of offloading scheduling, offloaded-data sizes, and communication-and-computation time division. By analyzing its properties,  we present a solution approach of decomposing the problem into master and slave sub-problems. The former optimizes the number of offloading users and given the number, the later optimizes offloading-user set, offloaded-data sizes, and time division (offloading, computing, downloading). By adopting Dinkelbach method,  an efficient iteratively algorithm is designed to solve the slave problem that is a \emph{combinatorial-optimization} problem. With the algorithm, the master problem can be then solved by a simple search over a finite integer set of possible numbers of offloading users. In addition, special cases are studied to yield useful design guidelines.

\item{\bf (Sum Mobile Energy Minimization)}{ The problem of sum-energy minimization is also non-convex. To develop practical scheduling algorithms for efficiently solving the problem, we divide the whole user set into multiple subsets based on the corresponding levels of offloading gain in terms of energy reduction. Then some reasonable rules are introduced to prioritize the user subsets' offloading so as to enable tractable algorithmic design. Based on the rules, an efficient greedy algorithm is proposed to schedule different subsets of users which achieves close-to-optimal performance as demonstrated by simulation.}
\end{itemize}

The rest of this paper is organized as follows. In Section II, we present our system model and problem formulation. In Section III, we propose an optimal algorithm to solve the problem of sum offloading throughput  maximization and discuss  special cases. The problem of  sum energy minimization  is studied in Section IV. Finally,  simulation results and conclusions are provided in Section V and Section VI, respectively.

\section{System Model and Problem Formulation}

\subsection{System Model}

Consider an MEC system shown in Fig. \ref{fig:subfig}, consisting of one AP integrated with an MEC server and $K$  users. Partial offloading is assumed in this paper so that each user can partition its computation task into two independent parts for  local computing and  offloading to the server. The two operations are simultaneous  as the communication modulars and user CPUs are separated. All of the users have to complete their tasks within a fixed duration $T$ (in second) so as to meet a real-time requirement. The system operation is divided into three \emph{sequential phases}: 1) TDMA-based task offloading by users, 2) parallel computing at the server, and 3) TDMA-based computation-result downloading from the server to users. Corresponding models and assumptions are described as follows.
%

\begin{figure}[t]
\begin{centering}
\includegraphics[width=0.7 \linewidth]{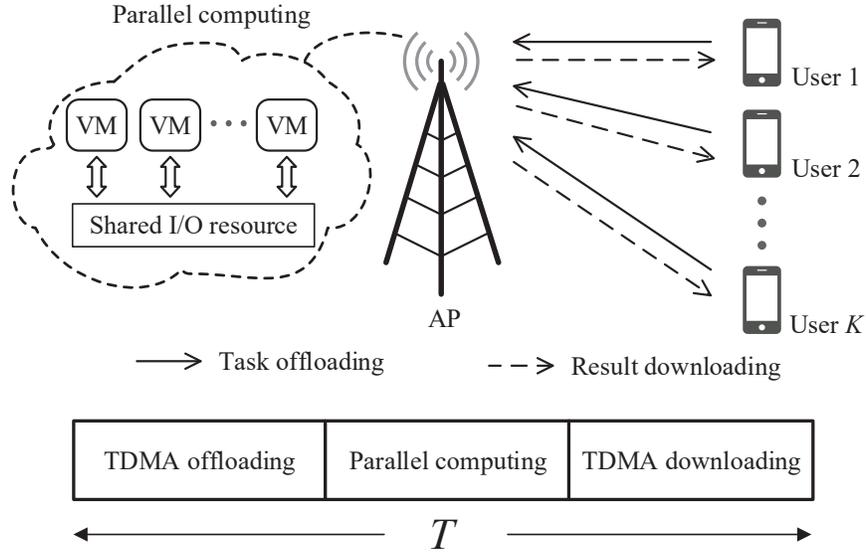}
\vspace{-0.2cm}
\caption{ A multiuser MEC system comprising a single AP and $K$ users.}\label{fig:subfig}
\end{centering}
\end{figure}

\subsubsection{Offloading and Downloading Phases}
Let $\ell_i$ denote the input data bits offloaded by user $i$ to the server. It is assumed that each input bit generates  $\gamma_i$ bits of computation result. Then for an offloaded data $\ell_i$, the computed result contains  $\gamma_i\ell_i$ bits. The  transmission   delay for  user $i$ for  offloading  and downloading can be written separately as
\begin{align}
t_{i}^u&=a_i \ell_i, \label{eqn:tiu}\\
t_{i}^d&=b_i \gamma_i \ell_i, \label{eqn:tid}
\end{align}
where $a_i$ and $b_i$ are the required time for transmitting a single bit in uplink and downlink, respectively, which are the inverse of the corresponding  uplink and downlink rates.

\subsubsection{Parallel-Computing Phase with Virtualization}
After receiving all the offloaded tasks, the server executes them in parallel by creating  multiple VMs. We consider the important factor of I/O interference in parallel computing \cite{netIO} and adopt a model developed in the literature based on measurement data \cite{parallel_computing_mode,survey}, which is described as follows.  Group the user indices into the set  $\mathcal{K}$. The subset $\mathcal{S}\subseteq\mathcal{K}$ identifies the set of scheduled offloading users, $t_e$ the time allocated to the parallel-computing phase, and $r_i$ the expected computation-service rate (bits/sec) of a VM given task $i$ when running in isolation.  Following  \cite{parallel_computing_mode,survey}, a performance degradation factor $d>0$  is defined to specify the percentage reduction in the computation-service rate of a VM when multiplexed with another VM. \footnote{The parameter $d$ depends on the specific VM multiplexing and
placement strategy \cite{degradation_factor1,degradation_factor2}. Its value  can be estimated by theoretical studies or statistical observations.} Suppose that one VM is created and assigned to a task, the degraded computing rate for each task is modeled as $r_i (1+d)^{1-|\mathcal{S}|}$ \cite{parallel_computing_mode}, where $|\mathcal{S}|$ denotes the number of tasks (or offloading users) for parallel computing. Therefore, for given $t_e$, the numbers of offloadable bits are constrained by
\begin{align}\label{eqn:st_ell}
0\leq \ell_i\leq t_e r_i (1+d)^{1-|\mathcal{S}|}, \quad \forall i\in \mathcal{S}.
\end{align}
The constraints in \eqref{eqn:st_ell} show that the maximum number of offloadable bits per user  decreases with the number of offloaded tasks due to the I/O interference in parallel computing. Moreover, relaxing the duration for parallel computing ($t_e$) can accommodate more offloaded bits ($\{\ell_i\}$), however, at the cost of less  time for  the offloading and downloading phases. This introduces a  tradeoff between the three phases under the following total-latency constraint:
\begin{align}
\sum_{i\in\mathcal{S}}t_i^u+t_e+\sum_{i\in\mathcal{S}}t_i^d =\sum_{i\in \mathcal{S}}\ell_i\left(a_i+b_i\gamma_i\right)+t_e \leq T. \label{eqn:T}
\end{align}

\subsection{Problem Formulation}
In this paper, we consider two popular system-performance metrics: sum offloading rate maximization and sum energy consumption minimization by users. The metrics target  two different scenarios where users are constrained in computing capacity and energy, respectively. Correspondingly, offloading aims at either enhancing user capacities or reducing their energy consumption. Using the metrics, two RCRA problems are formulated as follows.

\subsubsection{Sum Offloading Rate Maximization}
The objective is to maximize the weighted sum of the users' offloading rates by joint offloading-user scheduling, offloaded-bits control, and three-phase time allocation. Here, the sum offloading rate is defined as the sum offloadable bits over the time duration $T$. Let $\omega_i$ denote a positive weight assigned to user $i$ based on the users' priority. Mathematically, the optimization problem can be formulated as
\begin{subequations}
\begin{align}
(\text{P1}):~\max_{\mathcal{S}\subseteq \mathcal{K}, \{\ell_i\}, t_e}~~&R=\frac{1}{T}\sum_{i\in \mathcal{S}} \omega_i\ell_i  \label{eqn:Ob} \\
{\rm{s.t.}}~~~&\sum_{i\in \mathcal{S}}\ell_i\left(a_i+b_i\gamma_i\right)+t_e\leq T, \label{eqn:st_T}\\
&0\leq \ell_i \leq t_e r_i \left(1+d\right)^{1-|\mathcal{S}|},\quad \forall i \in \mathcal{S}. \label{eqn:st_ell1}
\end{align}
\end{subequations}

Problem (P1) is a mixed-integer programming problem comprising  both a combinatorial variable $\mathcal{S}$ and continuous variables $(\{\ell_i\}, t_e)$ and non-convex constraints \eqref{eqn:st_ell1}. Therefore, Problem (P1) is \emph{non-convex}. Though such a problem is usually difficult to solve exactly, an algorithm is designed in sequel to find the optimal solution.

\subsubsection{Sum Energy Minimization} We aim at minimizing the total energy consumed at all the users. Suppose that each user $i\in\mathcal{K}$ has a computation task of length $L_i$ bits, of which $\ell_i$ bits are  offloaded to the AP and $(L_i-\ell_i)$ bits computed locally. In parallel with offloading, the allowed duration for  local computing at any user  is $T$. Let the duration for user $i$ be denoted as $t_i^{\text{loc}}$. Then we have the following time constraint on local computing:
\begin{align}
t_i^{\text{loc}}&=\frac{c_i (L_i-\ell_i)}{f_i} \leq T,  \quad \forall i\in \mathcal{K}, \label{eqn:local_latency}
\end{align}
where $c_i$ denotes the fixed number of CPU cycles required to compute a single bit and $f_i$ denotes the CPU frequency at user $i$ (CPU cycles/sec). The  energy consumption for computing $(L_i-\ell_i)$  bits at user $i$ can be written as
\begin{align}
E_i^{\text{loc}}&=\kappa_i c_i (L_i-\ell_i) f_i^2,
\end{align}
where $\kappa_i$ is a coefficient depending on the specific hardware architecture. Combining constraints \eqref{eqn:st_ell} and \eqref{eqn:local_latency} and considering the fact that $0\leq\ell_i\leq L_i$ yield  the constraint on the number of offloadable bits as
\begin{align}
L_i^{\text{min}}\leq \ell_i \leq \min\left\{L_i, t_e r_i \left(1+d\right)^{1-|\mathcal{S}|}\right\},\quad \forall i \in \mathcal{S}, \label{eqn:st_ell2}
\end{align}
where $L_i^{\text{min}}\triangleq \left[L_i-\frac{Tf_i}{c_i}\right]^+$, with $[\cdot]^+\triangleq \max\{\cdot, 0\}$, is derived from \eqref{eqn:local_latency} by setting $t_i^{\text{loc}}=T$. On the other hand, the energy consumption for offloading $\ell_i$ bits is $E_i^{\text{off}}=t_i^u p_i=a_ip_i\ell_i$. Therefore, the total energy consumption of each user is $E_i=E_i^{\text{loc}}+E_i^{\text{off}}$. Then, the weighted sum energy consumption of all users can be expressed as
\begin{align} \label{eqn:Energy_objective}
E&=\sum_{i\in \mathcal{K}}\omega_i\left[E_i^{\text{loc}}+ E_i^{\text{off}} \right]\nonumber \\
&=\sum_{i\in\mathcal{K}}\omega_i\left[\kappa_ic_i\left(L_i-\ell_i\right)f_i^2+a_ip_i\ell_i\right] \nonumber \\
&=\sum_{i\in\mathcal{K}}\omega_i\left(a_ip_i-\kappa_ic_if_i^2\right)\ell_i+\sum_{i\in\mathcal{K}}\omega_i\kappa_ic_iL_if_i^2\nonumber \\ &=\sum_{i\in \mathcal{K}}\theta_i\ell_i + \rm{e}_0,
\end{align}
where $\theta_i\triangleq \omega_i\left(a_ip_i-\kappa_ic_if_i^2\right)$ and ${\rm{e}_0}\triangleq\sum_{i\in \mathcal{K}}\omega_i\kappa_ic_iL_if_i^2 $ are both constants.

Given the objective of minimizing the sum-energy consumption in \eqref{eqn:Energy_objective} subject to the time constraint in \eqref{eqn:T} and the offloadable bits constraints in \eqref{eqn:st_ell2}, the corresponding RCRA  problem can be formulated as
\begin{subequations}
\begin{align}
(\text{P2}):~\min_{\mathcal{S}\subseteq \mathcal{K}, \{\ell_i\}, t_e}~~&\sum_{i\in \mathcal{S}} \theta_i \ell_i \label{eqn:Ob2} \\
{\rm{s.t.}}~~~&\sum_{i\in \mathcal{S}}\ell_i\left(a_i+b_i\gamma_i\right)+t_e\leq T,  \label{eqn:st_T2}\\
&L_i^{\text{min}}\leq \ell_i \leq \min\left\{L_i, t_e r_i \left(1+d\right)^{1-|\mathcal{S}|}\right\},\quad \forall i \in \mathcal{S}, \label{eqn:st_S2}
\end{align}
\end{subequations}
where the objective function \eqref{eqn:Ob2} is derived from \eqref{eqn:Energy_objective} by omitting the constant term $\rm{e}_0$ and combining the fact that $\ell_i=0$ for all non-scheduled users.

Like Problem (P1), Problem (P2) is also \emph{non-convex}. An efficient algorithm is developed in the sequel to approximately solve this problem.

\section{Sum Offloading Rate Maximization}

In this section, we develop  an optimal algorithm for  solving  Problem (P1). First, an important property of the optimal offloading-user set $\mathcal{S}^*$ will be obtained, which allows tractable analysis of the  optimal offloading scheme and thereby simplifies  the problem. Subsequently, an  iterative algorithm based on the Dinkelbach method is  proposed to exactly solve the simplified problem. Last, we  discuss several special cases to obtain useful insights.

\subsection{Sum Offloading Rate Maximization for a  Given Offloading-User Set}

We  made a key observation  that Problem (P1) becomes a linear programming (LP) problem if the offloading-user set is given. The conditional optimal offloading strategy, specified by the offloaded data sizes  $\{\ell_i^*\}$,  satisfies  the following property.

\begin{lemma}\label{lemma1}
Given an arbitrary offloading-user set $\mathcal{S}$,  the optimal offloading strategy $\{\ell_i^*\mid i \in \mathcal{S}\}$  must be the maximum or minimum value in the constraint in \eqref{eqn:st_ell1}.
\end{lemma}
\begin{proof}
See Appendix A.
\end{proof}

Lemma \ref{lemma1} indicates that the optimal offloading strategy of each scheduled user follows a binary policy, i.e., offloading  with the maximum data size or  nothing. Accordingly, we can divide the elements in  $\{\ell_i\mid i \in \mathcal{S}\}$  into two groups, with one group $\widetilde{\mathcal{S}}$ for the users offloading maximum bits and the other performing no  offloading, i.e.,
\begin{align}\label{eqn:ell_S_tilde}
\ell_i&=
\begin{cases}
t_e r_i \left(1+d\right)^{1-|\mathcal{S}|}, &i\in\widetilde{\mathcal{S}},  \\
0, &i\in \mathcal{S}\backslash\widetilde{\mathcal{S}}.
\end{cases}
\end{align}
Note that $\widetilde{\mathcal{S}}$ is needed to be determined and we first use $\widetilde{\mathcal{S}}$ to express $\{\ell_i\}$ and $t_e$. It is intuitive that the equality must hold in constraint \eqref{eqn:st_T} for the optimal solution of Problem (P1). Then, by substituting  \eqref{eqn:ell_S_tilde} into \eqref{eqn:st_T}, we obtain the conditional optimal parallel-computing time $t_e$ as
\begin{align}
t_e=\frac{T(1+d)^{|\mathcal{S}|-1}}{(1+d)^{|\mathcal{S}|-1}+\sum_{i\in \widetilde{\mathcal{S}}}(a_i +b_i\gamma_i)r_i}. \label{eqn:t_e_S_tidle}
\end{align}
Combining \eqref{eqn:ell_S_tilde} and \eqref{eqn:t_e_S_tidle}, Problem (P1) for a given $\mathcal{S}$ can be formulated  as one for determining the subset $\widetilde{\mathcal{S}}$ in $\mathcal{S}$:
\begin{align}\label{eqn:tildeS}
\max_{\widetilde{\mathcal{S}}\subseteq\mathcal{S}}~~R=\frac{\sum_{i\in \widetilde{\mathcal{S}}}\omega_i r_i}{(1+d)^{|\mathcal{S}|-1}+\sum_{i\in \widetilde{\mathcal{S}}}(a_i+b_i\gamma_i)r_i}.
\end{align}
It is observed that, for any given $\mathcal{S}$, if $\widetilde{\mathcal{S}}^*\neq\mathcal{S}$, $R$ in \eqref{eqn:tildeS} can be further improved via replacing $\mathcal{S}$ with the smaller subset $\widetilde{\mathcal{S}}^*$. In other words, there exists the users $i\in \mathcal{S} \backslash\widetilde{\mathcal{S}}^*$ who offload zero bits but are scheduled to unnecessarily create a VM at the server, resulting in waste of resources. Thereby, removing them from the offloading-user set and only allocating VMs to the full offloading users can further increase the sum offloading rate. By the above argument, the necessary condition for $\mathcal{S}$ being optimal of Problem (P1) is $\widetilde{\mathcal{S}}^*=\mathcal{S}$ in Problem \eqref{eqn:tildeS}. That is, all the scheduled users offload their maximum bits, or otherwise the given $\mathcal{S}$ is not optimal.

Thus, we re-define $R$ as the sum offloading rate achieved by $\widetilde{\mathcal{S}}=\mathcal{S}$ in Problem \eqref{eqn:tildeS}, i.e.,
\begin{align}\label{eqn:R_revised}
R= \frac{\sum_{i\in \mathcal{S}}\omega_i r_i}{(1+d)^{|\mathcal{S}|-1}+\sum_{i\in \mathcal{S}}(a_i+b_i\gamma_i)r_i}.
\end{align}
We have the following proposition to identify whether $\mathcal{S}$ meets the necessarily optimal condition.
\begin{proposition}\label{proposition1}
$\widetilde{\mathcal{S}}=\mathcal{S}$ is the optimal solution of Problem \eqref{eqn:tildeS} if and only if the given offloading-user set $\mathcal{S}$ satisfies
\begin{align}
R\leq\min_{i\in \mathcal{S}}\left\{\frac{\omega_i}{a_i+b_i \gamma_i}\right\}. \label{eqn:S_pro1}
\end{align}
\end{proposition}
\begin{proof}
See Appendix B.
\end{proof}

To better understand \eqref{eqn:S_pro1}, we multiply  the term $t_e(1+d)^{1-|\mathcal{S}|}$ in both the numerator and denominator of $R$ and using the result that $\ell_i=t_e r_i (1+d)^{1-|\mathcal{S}|}$, $\forall i\in\mathcal{S}$, then $R$ in \eqref{eqn:R_revised} can be rewritten as
\begin{align}\label{eqn:RS_understanding}
R&=\frac{\sum_{i\in \mathcal{S}}\omega_i \ell_i}{t_e+\sum_{i\in\mathcal{S}}(a_i+b_i\gamma_i)\ell_i}.
\end{align}
The numerator in  \eqref{eqn:RS_understanding} denotes the sum offloaded bits and the denominator denotes the total time and equals $T$. Then, $R$ in \eqref{eqn:RS_understanding} can be physically interpreted as the sum offloading rate of the system with users of set $\mathcal{S}$. On the other hand, $\frac{\omega_i}{a_i+b_i \gamma_i}$ can be rewritten as $\frac{\omega_i \ell_i}{(a_i+b_i\gamma_i)\ell_i}$,
where the numerator denotes the offloaded bits of user $i$ while the denominator denotes the transmission duration that includes both offloading and downloading time. Therefore, $\frac{\omega_i}{a_i+b_i \gamma_i}$ can be regarded as the transmission rate of user $i$. Proposition \ref{proposition1} implies that the system offloading rate should be less than or equal to the minimum transmission rate among users in $\mathcal{S}$ if it solves Problem (P1).
\begin{remark}\label{remark1}
If the given offloading-user set $\mathcal{S}$ violates condition \eqref{eqn:S_pro1}, $R$ can be further improved by removing users with minimum transmission rate from the offloading-user set.
\end{remark}
Let index $j$ denote the user with minimum transmission rate in set $\mathcal{S}$. Remark \ref{remark1} can be illustrated using the following inequality:
\begin{align}\label{eqn:S_remark1_inequ}
\frac{\omega_j\ell_j}{(a_j+b_j\gamma_j)\ell_j}<\frac{\sum_{i\in \mathcal{S}}\omega_i \ell_i}{t_e+\sum_{i\in\mathcal{S}}(a_i+b_i\gamma_i)\ell_i}< \frac{\sum_{i\in \mathcal{S}\backslash\{j\}}\omega_i \ell_i}{t_e+\sum_{i\in\mathcal{S}\backslash\{j\}}(a_i+b_i\gamma_i)\ell_i},
\end{align}
where the middle term in \eqref{eqn:S_remark1_inequ} is identical to $R$ and the right hand side of \eqref{eqn:S_remark1_inequ} is the sum offloading rate achieved after removing user $j$. \eqref{eqn:S_remark1_inequ} reveals that when the given  offloading-user set $\mathcal{S}$ violates condition \eqref{eqn:S_pro1}, there exists a slow user (i.e., user $j$) that is a bottleneck in the transmission process. Even without accounting for the parallel computing time, its transmission rate is already slower than the system offloading rate. Therefore, removing this bottleneck user can further improve the system offloading rate.

\subsection{Offloading-User Scheduling}
Building on the results from  the last subsection, we present in this subsection an efficient scheduling algorithm for computing the optimal offloading-user set.  To this end, the variables $\{\ell_i\}$ and $t_e$ can be expressed in term of $\mathcal{S}$ when $\mathcal{S}$ meets the necessarily optimal condition \eqref{eqn:S_pro1}. This simplifies Problem (P1) as  a scheduling problem that finds the optimal offloading-user set under constraint \eqref{eqn:S_pro1}:
\begin{align}\label{eqn:problem_of_S}
\max_{\mathcal{S}\subseteq\mathcal{K}} \quad &R=\frac{\sum_{i\in \mathcal{S}}\omega_i r_i}{(1+d)^{|\mathcal{S}|-1}+\sum_{i\in \mathcal{S}}(a_i+b_i\gamma_i)r_i} \\
{\rm{s.t.}}~~ &R\leq\min_{i\in \mathcal{S}}\left\{\frac{\omega_i}{a_i+b_i \gamma_i}\right\}. \nonumber
\end{align}
The problem can be further reduced to an unconstrained optimization problem using the following useful result.
\begin{proposition}\label{proposition2}
Constraint \eqref{eqn:S_pro1} can be removed from Problem \eqref{eqn:problem_of_S} without loss of optimality.
\end{proposition}
\begin{proof}
See Appendix C.
\end{proof}

Using Proposition \ref{proposition2}, Problem \eqref{eqn:problem_of_S} can be safely relaxed into the following  non-constrained optimization problem:
\begin{align}\label{eqn:problem_of_S1}
\max_{\mathcal{S}\subseteq\mathcal{K}} \quad&R=\frac{\sum_{i\in \mathcal{S}}\omega_i r_i}{(1+d)^{|\mathcal{S}|-1}+\sum_{i\in \mathcal{S}}(a_i+b_i\gamma_i)r_i}.
\end{align}

However, with the non-convex term $(1+d)^{|\mathcal{S}|-1}$ in the denominator of $R$, Problem \eqref{eqn:problem_of_S1} is still challenging to solve. To tackle this difficulty, we fix $|\mathcal{S}|=m$, with $m=1,\cdots, K$. For a given $m$, since term $(1+d)^{|\mathcal{S}|-1}$ becomes a constant, Problem \eqref{eqn:problem_of_S1} is reduced to a \emph{mixed-integer linear fractional programming} problem. We solve Problem \eqref{eqn:problem_of_S1} by decomposing it into master-and-slave problems without loss of the optimality. The slave problem is determining the optimal offloading-user set using the Dinkelbach method \cite{NonLinear} for a given number of scheduled users $m$. Then the master problem is obtaining the optimal value of $m$, denoted as $m^*$,  by a simple search. The detailed solutions of the decomposed problems are presented in the sequel, which yield Algorithm~1 for computing the optimal scheduled-user set  $\mathcal{S}^*$.

\subsubsection{Optimal scheduling for a given number of scheduled users}
In this section, we solve Problem (P1) conditioned on a given  number of offloading users $m$, i.e., $|\mathcal{S}|= m$. To this end, we introduce a set of binary variables $\mathbf{x}=[x_1,\cdots,x_K]$, where $x_i=1$ means that user $i$ is scheduled (i.e., $i \in\mathcal{S}$),  and $x_i=0$ otherwise. Then, using the binary variables and conditioned on $|\mathcal{S}|= m$, Problem (P1) can be transformed into a combinatorial optimization problem as
\begin{align}
\text{(Slave Problem)} \quad \begin{aligned}
\max_{\mathbf{x}}~~&R_m=\frac{\sum_{i=1}^{K}x_i\omega_i r_i }{(1+d)^{m-1}+\sum_{i=1}^{K}x_i(a_i+b_i\gamma_i)r_i}=\frac{N(\mathbf{x})}{D(\mathbf{x})} \label{eqn:subproblem_of_S}\\
{\rm{s.t.}}~~&\sum_{i=1}^{K}x_i= m,  \quad x_i\in \{0, 1\},  \quad i=1,\cdots,K,
\end{aligned}
\end{align}
where $N(\mathbf{x})\triangleq\sum_{i=1}^{K}x_i\omega_i r_i$ and $D(\mathbf{x})\triangleq(1+d)^{m-1}+\sum_{i=1}^{K}x_i(a_i+b_i\gamma_i)r_i$. Let $R_m^*$ denotes the maximum conditional sum offloading rate from solving the slave problem. For ease of notation, we define the feasible set for Problem \eqref{eqn:subproblem_of_S} as
$\mathcal{F}_m\triangleq\{\mathbf{x}|\sum_{i=1}^{K}x_i= m \text{ and }x_i\in \{0,1\}, i=1,\cdots,K\}$. Since the objective function has a fractional form, the problem can be solved by \emph{non-linear fractional programming}. To this end, define a function $g(\cdot)$ of the conditional rate $R_m$ by an optimization problem in  a substrative form:
\begin{align}\label{eqn:subtractive_of_Sm}
g(R_m)=\max_{\mathbf{x}\in \mathcal{F}_m}\left[N(\mathbf{x})- D(\mathbf{x})R_m\right].
\end{align}
Let $\mathbf{x}^*$ be an optimal solution of Problem \eqref{eqn:subproblem_of_S}. We have the following property.
\begin{lemma}\label{lemma2}
The maximum conditional sum offloading rate $R_m^*$ that solves Problem \eqref{eqn:subproblem_of_S} can be achieved if and only if
\begin{align}\label{eqn:equivalent_of_Sm}
g(R^*_m)=\max_{\mathbf{x}\in \mathcal{F}_m}\left[N(\mathbf{x})-D(\mathbf{x})R^*_m\right]=N(\mathbf{x}^*)-D(\mathbf{x}^*)R^*_m=0.
\end{align}
\end{lemma}
\begin{proof}
See Appendix D.
\end{proof}

Lemma \ref{lemma2} reveals the fact that the targeted fractional-form problem in \eqref{eqn:subproblem_of_S} shares the solution $\mathbf{x}^*$ as the subtractive-form problem in  \eqref{eqn:subtractive_of_Sm} when $R_m=R_m^*$. This provides an indirect method for solving the former using an iterative algorithm derived in the sequel, in which the derived condition $g(R_m) = 0$ is applied to checking the optimal convergence.

Based on Dinkelbach method \cite{NonLinear}, we propose an iterative algorithm to obtain $R_m^*$ in \eqref{eqn:equivalent_of_Sm}, thereby solving the slave problem in \eqref{eqn:subproblem_of_S}. Specifically, we concern the optimal solution to the subtractive-form Problem \eqref{eqn:subtractive_of_Sm} for a given $R_m$:
\begin{align}\label{eqn:subproblem_of_Xm}
g(R_m)&=\max_{\mathbf{x}\in \mathcal{F}_m}\left\{\sum_{i=1}^{K} x_ir_i \left[\omega_i-R_m(a_i+b_i\gamma_i)\right]-R_m(1+d)^{m-1}\right\}.
\end{align}

To facilitate exposition, we can rewrite the expression of $g(R_m)$ as
\begin{align}
g(R_m)&=\max_{\mathbf{x}\in \mathcal{F}_m}\left\{\sum_{i=1}^{K} \frac{x_i\left[\omega_i t_er_i(1+d)^{1-m}-R_m (a_i+b_i\gamma_i)t_er_i(1+d)^{1-m}\right]}{t_e(1+d)^{1-m}}-R_m(1+d)^{m-1}\right\}\nonumber\\
&=\max_{\mathbf{x}\in \mathcal{F}_m}\left\{ \frac{\sum_{i=1}^{K}x_i\psi_i(R_m)
}{t_e(1+d)^{1-m}} -R_m(1+d)^{m-1}\right\},\label{eqn:subproblem_of_Xm2}
\end{align}
where the last equality is obtained by substituting \eqref{eqn:tiu} and \eqref{eqn:tid} and defining
\begin{equation}\label{eqn:remark2}
\psi_i(R_m)=\omega_i\ell_i-R_m(t_i^u+t_i^d).
\end{equation}

\begin{remark}[Per-user Revenue]\label{Remark2}
The variable $\psi_i(R_m)$ can be interpreted as the net revenue of scheduling user $i$ as explained shortly. With the system offloading rate $R_m$, $R_m(t_i^u+t_i^d)$ represents  the expected number of user bits that can be computed successfully by offloading and result downloading over the duration  of $(t_i^u+t_i^d)$. By allocating the time to  user $i$ for offloading and downloading, the weighted number of actual  computed bits  is $w_i\ell_i$. Therefore, the difference between expected and actual bits, $\psi_i(R_m)$, measures the net system revenue obtained from scheduling user $i$.
\end{remark}

\begin{itemize}
\item{\bf Step 1}: Based on Remark \ref{Remark2}, the objective of the optimization in \eqref{eqn:subproblem_of_Xm} can be interpreted as one for maximizing the total system revenue. It follows that the optimal solution, denoted as $\mathbf{x}^{*}$, is to select $m$ users having the largest per-user revenue:
        \begin{align}\label{eqn:xi_op}
        x_{i}^*=
        \begin{cases}
        1, &\text{if $\psi_i(R_m)$ is one of the $m$ largest},\\
        0, & \text{otherwise},
        \end{cases}
        \end{align}
    with $i=1,\cdots,K$, where $\psi_i(R_m)$ is defined in \eqref{eqn:remark2}.

\item {\bf Step 2}: Given $\mathbf{x}^*$ computed in Step 1, the sum offloading rate $R_m$ can be updated as
\begin{align}\label{eqn:update_delta}
R_m=\frac{N(\mathbf{x}^*)}{D(\mathbf{x}^*)},
\end{align}
where $N(\cdot)$ and $D(\cdot)$ are given in \eqref{eqn:subproblem_of_S}. Then the per-user revenues $\{\psi_i(R_m)\}$ are updated using the new value of $R_m$.

\end{itemize}

Based on the Dinkelbach method, the above two steps are iterated till $g(R_m)=0$. Since this is the optimality condition according Lemma \ref{lemma2}, the convergence of the iteration yields the maximum $R_m^*$ and the corresponding $m$ scheduled users $\mathcal{S}^*(m) = \{i\mid x_i^* = 1\}$. It can be proved that the convergence rate is \emph{superlinear} (see e.g., \cite{NonLinear}).

\subsubsection{Finding the optimal number of scheduled users} With the slave problem in \eqref{eqn:subproblem_of_S} solved in the preceding sub-section, the master problem is to optimize $m$:
\begin{align}
\text{(Master Problem)}\qquad \max_{1\leq m\leq K}\quad R^*_m=\frac{\sum_{i\in \mathcal{S}^*(m)}\omega_i r_i}{(1+d)^{m-1}+\sum_{i\in \mathcal{S}^*(m)}(a_i+b_i\gamma_i)r_i}.\label{eqn:masterproblem}
\end{align}

To solve the problem, an intelligent search for $m^*$ over $\{1, 2, \cdots,  K\}$ seems to be difficult for the reason that $\{R_m^*\}$ is not a monotone sequence, which arises from the fact that scheduling more users increases multiplexing gain in parallel computing but causes stronger I/O interference and vice versa. Due to the lack of monotonicity, we resort to enumerating all possible values of $m$ from $1$ to $K$ to find $m^*$. The complexity of the exhaustive search is reasonable as it scales only linearly with the total number of users $K$.

\subsubsection{Overall Algorithm and Its Complexity}
The overall algorithm for solving the scheduling problem in \eqref{eqn:problem_of_S1} is shown in  Algorithm 1 which combines the iterative algorithm for solving the slave problem and the exhaustive search for solving in the master problem which are designed in the  preceding sub-sections.

The complexity of the overall algorithm is discussed as follows. The iterative algorithm for solving the slave problem using Dinkelbach  method has complexity upper bounded by $O(\log K)$ \cite{matsui1992analysis}.  Solving the master problem repeats at most $K$ runs of the iterative algorithms. Therefore, the worst-case complexity of the overall algorithm is $O(K\log K)$.

\begin{algorithm}[t] \label{alg:1}
\caption{Iterative User Scheduling Algorithm Based on Dinkelbach Method}
\begin{algorithmic}[1]
\FOR{ $m=1,\cdots$, $K$}
\STATE \textbf{initialize} $R_m=0$.
\REPEAT
\STATE For a given $R_m$, compute $\mathbf{x}^{*}$ according to \eqref{eqn:xi_op};
\STATE Update $R_m=\frac{N(\mathbf{x}^*)}{D(\mathbf{x}^*)}$;
\UNTIL $g(R_m)=0$.
\STATE Return $R^*_m=R_m$, $\mathbf{x}^*_m=\mathbf{x}^{*}$.
\ENDFOR
\STATE Return $m^*=\arg\max_{1\leq m \leq K} \left\{R^*_m\right\}$,  $R^*=R^*_{m^*}$ and $\mathcal{S}^*=\left\{ i~|~ x_{i}^*=1, i\in \mathcal{K}\right\}$.
\ENSURE $R^*$ and $\mathcal{S}^*$.
\end{algorithmic}
\end{algorithm}

\subsection{Special Cases}
Several special cases are considered to derive additional insights into the optimal multiuser offloading. For simplicity, the users' weights are assumed to be uniform, i.e., $\omega_i=1$, $\forall i\in \mathcal{K}$.

\subsubsection{Homogenous Users and Channels} Consider the special case where users are homogeneous in task types and channels such that their offloading parameter sets $\{r_i, a_i, b_i, \gamma_i\}$ are identical.
Then Problem \eqref{eqn:problem_of_S1} reduces to the simple problem of determining the number of offloading users. The sum offloading rate $R$ can be simplified as $R=\frac{mr}{(1+d)^{m-1}+mr(a+b\gamma)}$. By letting $\frac{d R}{d m}=0$,  the optimal number of scheduled users is obtained as
\begin{align}\label{eqn:opti_in_uniform_case}
m^*\approx \left[\frac{1}{\ln (1+d)}\right]^K_1,
\end{align}
where $\left[x\right]^{K}_{1}=\max\left\{\min\left\{x, K\right\}, 1\right\}$ restricts the $m^*$ in the range from $1$ to $K$. The result shows that for this special case, the optimal number of offloading users (or equivalently the optimal number of VMs in parallel computing) only depends on the I/O-interference parameter  $d$ in the parallel-computing model in \eqref{eqn:st_ell}.

\subsubsection{Homogeneous  Transmission Rates} Relaxing the assumption of homogeneous task types in the preceding case leads to the current case of homogeneous transmission rates due to channel homogeneity, corresponding to $\frac{1}{a_1+b_1\gamma_1}= \frac{1}{a_2+b_2\gamma_2}=\cdots = \frac{1}{a_K+b_K\gamma_K}$. Due to variation in task types, users have different computation-service rate specified by the parameter $\{r_i\}$ in the computation model in \eqref{eqn:st_ell}. We obtain for the current case the optimal offloading-user set as shown in the following proposition.
\begin{proposition}[Homogeneous Transmission Rates]\label{proposition4}
Consider the special case of homogeneous transmission rates $\frac{1}{a_1+b_1\gamma_1}= \frac{1}{a_2+b_2\gamma_2}=\cdots = \frac{1}{a_K+b_K\gamma_K}$. Without loss of generality, assume the computation-service rates   $r_1\geq r_2\geq...\geq r_K$. Let  $n_0$ denote the largest user index $n$ that satisfies $r_n\geq d\sum_{i=0}^{n-1}r_i$
with  $r_0=0$. Then, the optimal scheduled-user set $\mathcal{S}^*$ that solves Problem \eqref{eqn:problem_of_S} is given by
\begin{equation}
\mathcal{S}^*=\left\{i~| 1\leq i\leq n_0\right\}.
\end{equation}
 \end{proposition}

\begin{proof}
Please see Appendix E.
\end{proof}

\begin{remark}[To schedule or not?]
The computation-service rate $r_n$ can be seen as the gain of scheduling  user $n$ while $d\sum_{i=0}^{n-1}r_i$ represents the performance degradation imposed on the preceding  scheduled users (i.e., user $1$ to $(n-1)$). As long as $r_n\geq d\sum_{i=0}^{n-1}r_i$ is met, the gain of scheduling user $n$ outweighs its cost and thus it is worthwhile to schedule user $n$ for improving the sum offloading rate.
\end{remark}

\begin{remark}[Optimal Scheduling]
Proposition~\ref{proposition4} shows that the optimal scheduling policy is to select $n_0$ users with the best computing rates and the index $n_0$ can be obtained by adopting greedy approach that selects users in descending order of the computation-service rate (i.e., $r_i$) until the condition $r_n\geq d\sum_{i=0}^{n-1}r_i$ becomes invalid.
\end{remark}

\subsubsection{No I/O Interference}
Consider the ideal case without I/O interference, namely $d=0$ in \eqref{eqn:st_ell}. This case corresponds to sufficient I/O resources at the AP. We can show that for this case the optimal offloading-user set has a threshold based structure where the threshold is determined by transmission rates. The details are given in the following proposition.

\begin{proposition}\label{proposition5}
Without loss of generality, assume the transmission rates follow the descending order: $\frac{1}{a_1+b_1\gamma_1}> \cdots > \frac{1}{a_K+b_K\gamma_K}$. Let $m_0$, with $1\leq m_0\leq K$, denote the largest user index that meets
\begin{align}\label{eqn:ordering}
\frac{1}{a_{m_0}+b_{m_0}\gamma_{m_0}}\geq \frac{\sum_{i=1}^{m_0}r_i}{1+\sum_{i=1}^{m_0}(a_i+b_i\gamma_i)r_i}.
\end{align}
The optimal scheduled-user set $\mathcal{S}^*$ that solves Problem  \eqref{eqn:problem_of_S} is given by
\begin{equation}
\mathcal{S}^*=\left\{i~| 1\leq i\leq m_0\right\}.\label{Eq:OpScheduling}
\end{equation}
\end{proposition}
The proof is similar to Proposition \ref{proposition4} and thus omitted. One can observe that condition \eqref{eqn:ordering} is a simplified version of \eqref{eqn:S_pro1} by setting $d=0$. However, it is important to note that the former provides a sufficient condition of the optimal offloading-user set in \eqref{Eq:OpScheduling} for the current special case while the latter only provides a necessary condition for optimal scheduling in the general case. Last, similar to the preceding special case, the index $m_0$ can be obtained via a greedy method.

\section{Sum Mobile Energy Minimization}
In this section, we attempt to solve Problem (P2) of minimizing sum-energy consumption over mobiles in the multiuser-offloading process. First, the feasibility region of the problem is analyzed. Then Problem (P2) is converted into an equivalent problem, which facilitates the design of an algorithm for finding a sub-optimal solution.

\subsection{Feasibility Analysis}
The feasible region of Problem (P2) is non-empty if the latency constraint $T$ is larger than  the minimum required time for computing all offloaded  tasks, denoted as $T_\text{min}$. To find $T_\text{min}$ is equivalent to solving the following \emph{latency minimization problem}:
\begin{align*}
(\text{P3}):~ \min_{\mathcal{S}\subseteq\mathcal{K}, {\{\ell_i\}}, t_e, T} \quad&~ T  \\
{\rm{s.t.}}\quad~~~&\sum_{i\in \mathcal{S}}\ell_i\left(a_i+b_i\gamma_i\right)+t_e\leq T,  \\
&L_i^{\text{min}}\leq \ell_i \leq \min\left\{L_i, t_e r_i \left(1+d\right)^{1-|\mathcal{S}|}\right\},\quad \forall i \in \mathcal{S}.
\end{align*}
The solution of Problem (P3) can be obtained as shown in the following proposition.
\begin{proposition}\label{proposition6}
The minimum computation time $T_{\text{min}}$ that solves Problem (P3) is the root of the following equation with the variable $T$:
\begin{align}\label{eqn:ECM_feasi2}
\sum_{i\in \mathcal{K}}\left(a_i+ b_i \gamma_i\right)L_i^{\emph{min}} +\max_{i\in \mathcal{K}}\left\{\frac{L_i^{\emph{min}}}{r_i(1+d)^{1-N(T)}}\right\}=T,
\end{align}
where $N(T)\triangleq\sum_{i\in \mathcal{K}}\mathds{1}\{L_i^{\emph{min}}>0\}$ with $L_i^{\emph{min}}$ being the minimum offloaded data size under the latency constraint, and $\mathds{1}\{\cdot\}$ is an indicator function that outputs $1$ when an  event occurs and $0$ otherwise.
\end{proposition}
\begin{proof}
Please see Appendix F.
\end{proof}
The left hand side of \eqref{eqn:ECM_feasi2} represents the time required for completing the minimum offloaded data with sizes of $\left\{L_i^{\text{min}}\right\}$ with  $L_i^{\text{min}}=\left[L_i-\frac{Tf_i}{c_i}\right]^+$ [see \eqref{eqn:st_ell2}]. In this expression, the first term is the time used for offloading and downloading, and the second term is the time for parallel computing, which is dominated by the task with the maximum execution time. Note that minimum offloading with sizes  $\left\{L_i^{\text{min}}\right\}$ requires full utilization of   local-computing capacities, i.e., the local computing time is at its maximum extended to $t_i^{\text{loc}}=T$, for all $i$. It follows that  $T_\text{min}$ occurs when the time for local computing and MEC are both equal to $T$. Since the left hand side of \eqref{eqn:ECM_feasi2} is non-differentiable but monotonically decreasing with $T$, $T_\text{min}$ can be easily found by simple bisection search. Then $T\geq T_{\min}$ yields the condition of nonempty feasibility region for Problem (P2).

\subsection{Problem Transformation}
Problem (P2) can be transformed into an equivalent problem whose solution facilitates scheduling design. A close observation of Problem (P2) reveals  that, when the $i$-th minimum offloaded data size $L_i^{\text{min}}>0$, it indicates that task $i$  cannot  be computed locally within the duration $T$ and a fraction  with at least  $L_i^{\text{min}}$ bits has to be offloaded. Therefore, $L_i^{\text{min}}>0$ means that user $i$ needs offloading. On the other hand, the condition $\theta_i>0$ corresponds to the case where task offloading consumes more energy than local computing. For the purpose of energy-saving, users with $\theta_i>0$ should  offload the minimum of  $L_i^{\text{min}}$ bits. Based on if none, one, or both   of  the above two conditions holds,  we can divide $K$ users into four disjoint subsets as follows:
\begin{align*}
\mathcal{M}_0 &=\{i~ |~ L_i^\text{min}>0\text{ and }~ \theta_i>0\}, \quad ~~ \mathcal{N}_0 =\{i~ |~ L_i^\text{min}=0\text{ and }~ \theta_i>0\},\\
\mathcal{M}_1 &=\{i~ |~ L_i^\text{min}>0\text{ and }~ \theta_i<0\},\quad ~~ \mathcal{N}_1 =\{i~ |~ L_i^\text{min}=0\text{ and }~ \theta_i<0\}.
\end{align*}
As a result, $\mathcal{M}_0$ and $\mathcal{M}_1$ are the sets of users requiring offloading under the latency constraint. To save energy, users in $\mathcal{M}_0$ should  offload minimum data $\ell_i^*= L_i^{\text{min}}$  while users in $\mathcal{N}_0$ should perform local computing only (i.e., $\ell_i^*=0$). The other sets  $\mathcal{M}_1$ and $\mathcal{N}_1$ are the sets of users who favour offloading since it is more energy-efficient than local computing. Furthermore,  users in $\mathcal{M}_1$ have to offload at least  $L_i^\text{min}$ bits under the latency constraint. In summary, for sum energy minimization, the optimal scheduling policy should schedule all users in $\mathcal{M}_0$ with minimum offloading, all in $\mathcal{M}_1$ to offload at least $L_i^{\text{min}}$ bits, none from $\mathcal{N}_0$, a subset of users from $\mathcal{N}_1$ with nonzero offloading.

Based on the above discussion and by denoting an arbitrary subset of $\mathcal{N}_1$ as  $\mathcal{S}_1$,  Problem (P2)  can be transformed into the following  equivalent problem:
\begin{align}
(\text{P4}):~\min_{\mathcal{S}_1\subseteq \mathcal{N}_1, \{\ell_i\}, t_e, } &\quad \sum_{i\in \mathcal{S}_1\cup \mathcal{M}_1} \theta_i \ell_i \label{eqn:P2}\\
{\rm{s.t.}} &\quad \sum_{i\in \mathcal{S}_1\cup \mathcal{M}_1}\ell_i \left(a_i+ b_i \gamma_i\right)+t_e\leq \widetilde{T},\label{eqn:Const}\\
&\quad L_i^\text{min} \leq \ell_i \leq \min \left\{L_i, t_e r_i \left(1+d\right)^{1-|\mathcal{M}|-|\mathcal{S}_1|}\right\},\quad \forall i \in \mathcal{M}_1  \label{eqn:ECM_ell_M1}\\
&\quad 0\leq \ell_i \leq \min  \left\{L_i, t_e r_i \left(1+d\right)^{1-|\mathcal{M}|-|\mathcal{S}_1|}\right\},\quad \forall i\in \mathcal{S}_1 \label{eqn:ECM_ell_N1}\\
&\quad t_e \geq \max_{i\in \mathcal{M}_0}\left\{\frac{L_i^\text{min}}{r_i(1+d)^{1-|\mathcal{M}|-|\mathcal{S}_1|}}\right\}. \label{eqn:te_M0}
\end{align}
where $\widetilde{T}\triangleq T-\sum_{i\in \mathcal{M}_0} L_i^\text{min}(a_i+ b_i \gamma_i)$ and $|\mathcal{M}|\triangleq|\mathcal{M}_0|+|\mathcal{M}_1|$.  Note that $\{\theta_i <0 |  i\in \mathcal{S}_1 \cup \mathcal{M}_1\}$, corresponding to the fact  that offloading saves mobile energy.  It follows that the minimization in Problem (P4) attempts to maximize offloading for users from $\mathcal{N}_1$ and $\mathcal{M}_1$. Last, given $\mathcal{S}_1$, the total number of offloading users is obtained as $|\mathcal{M}_0| + |\mathcal{M}_1| + |\mathcal{S}_1|$. The constraint in \eqref{eqn:te_M0} is derived from the constraint \eqref{eqn:st_ell2} for $i\in \mathcal{M}_0$, which ensures that $t_e$ is no less than the minimum required time for computing any  task in $\mathcal{M}_0$.

Problem (P4) is a mixed integer programming problem. Its solution potentially  requires an exhaustive search over all possible user subsets $\{\mathcal{S}_1\}$, resulting in complexity exponentially increasing with number of  users. For this reason, we find a close-to-optimal solution by developing a  low-complexity suboptimal algorithm in the next subsection.

\subsection{Suboptimal Scheduling Algorithm}
The tractability of algorithmic design relies on  applying a set of sub-optimal rules on offloading so as to ensure that the latency requirement can be met. To this end, we make the observation that, without considering the latency constraint in  \eqref{eqn:Const}, the optimal offloading policy for solving  Problem (P4) is one that all the users in $\mathcal{M}_1$ and $\mathcal{N}_1$ offload data with  maximum sizes, i.e., $\{\ell_i=L_i | i\in \mathcal{M}_1\cup \mathcal{N}_1\}$. To rein in the latency,
 the following rules are proposed to simplify the design problem:
\begin{enumerate}
\item{The users in $\mathcal{N}_1$ are constrained to adopt binary offloading scheme, i.e., $\ell_i=L_i$ if $i\in \mathcal{S}_1$ and $\ell_i=0$ if $i\in\mathcal{N}_1\backslash \mathcal{S}_1$. The incentive of considering binary policy is to make the process of offloading-decision making as simple and efficient as possible.}
\item{If $\mathcal{S}_1\neq \emptyset$, the users in $\mathcal{M}_1$ offload their data with the maximum sizes $L_i$. The rule is motivated by the observation that increasing the number of simultaneous VMs incurs higher server's operational cost. Applying this rule can maximize the utilization of each subscribed VM resource so as to reduce the number of VMs for minimizing the cost, while ensuring a high system performance.}
\item{By applying the rules in 1) and 2), if the latency requirement \eqref{eqn:Const} is violated, we first remove the users in $\mathcal{S}_1$ for reducing  the total latency. When $\mathcal{S}_1=\emptyset$ and the reduced total latency still violates the requirement, we proceed  to reduce the offloading bits from  users in $\mathcal{M}_1$ for further latency reduction. }
\end{enumerate}

Based on the above rules, a sub-optimal algorithm for solving Problem (P4) is designed as follows. To this end, we define a function of total offloaded-computation latency for scheduled users as follows
\begin{align}\label{eqn:total_delay}
\mathcal{D}\left(\mathcal{S}_1\right)=\sum_{i\in  \mathcal{S}_1\cup \mathcal{M}_1} L_i (a_i+b_i\gamma_i)+\sum_{i\in \mathcal{M}_0} L_i^\text{min} (a_i+b_i\gamma_i)+t_e(\mathcal{S}_1),
\end{align}
with
\begin{align}\label{eqn:te_SN1}
t_e\left(\mathcal{S}_1\right)=\max\left\{\max_{i\in \mathcal{S}_1\cup \mathcal{M}_1}\left\{\frac{L_i}{r_i}\right\}, \max_{i\in \mathcal{M}_0}\left\{\frac{L_i^\text{min}}{r_i}\right\}\right\} (1+d)^{|\mathcal{M}|+|\mathcal{S}_1|-1},
\end{align}
Since all users in $\mathcal{M}_1\cup \mathcal{M}_0$ should be all scheduled as discussed earlier, the function has only one variable  $\mathcal{S}_1$.  The variable $t_e\left(\mathcal{S}_1\right)$ in \eqref{eqn:te_SN1} represents the corresponding minimum parallel computing time, which is derived using \eqref{eqn:ECM_ell_M1} to \eqref{eqn:te_M0}.

\begin{algorithm}[t] \label{alg:2}
\caption{Suboptimal Algorithm for solving Problem (P4)}
\begin{algorithmic}[1]
\REQUIRE $T\geq T_\text{min}$.
\STATE Set $\{\ell_i=0 |  i\in \mathcal{N}_0\}$ and $\{\ell_i=L_i^\text{min} | i\in \mathcal{M}_0\}$.
\STATE \textbf{initialize} $\{\ell_i=L_i | i\in \mathcal{M}_1\cup\mathcal{N}_1\}$.
\IF{$ T \in \left[\mathcal{D}(\mathcal{N}_1), +\infty\right)$}
\STATE Return $\mathcal{S}=\mathcal{M}_0\cup\mathcal{M}_1\cup\mathcal{N}_1$, $\{\ell_i\}$, and $t_e=t_e(\mathcal{N}_1)$.
\ELSIF{$T \in \left[\mathcal{D}(\emptyset), \mathcal{D}(\mathcal{N}_1)\right)$}
\STATE \textbf{initialize} $\mathcal{S}_1=\mathcal{N}_1$.
\REPEAT
\STATE Let $j=\arg\min_{i\in\mathcal{S}_1}\left\{\frac{-\theta_i}{a_i+b_i\gamma_i}\right\}$. Update $\mathcal{S}_1=\mathcal{S}_1\backslash\{j\}$ and $\ell_j=0$.
\UNTIL $\mathcal{D}(\mathcal{S}_1)\leq T$.
\STATE Return $\mathcal{S}=\mathcal{M}_0\cup\mathcal{M}_1\cup\mathcal{S}_1$, $\{\ell_i\}$, and $t_e=t_e(\mathcal{S}_1)$.
\ELSE
\STATE Solve Problem \eqref{eqn:LP_M1} and obtain its optimal solution $\{\ell_i^\prime\}_{i\in \mathcal{M}_1}$ and $t_e^\prime$.
\STATE Set $\{\ell_i=0 |  i\in \mathcal{N}_1\}$, $\{\ell_i=\ell_i^\prime| i\in \mathcal{M}_1\}$, and $t_e=t_e^\prime$. \STATE Return $\mathcal{S}=\mathcal{M}_0\cup\mathcal{M}_1$, $\{\ell_i\}$, and $t_e$.
\ENDIF
\ENSURE $\mathcal{S}$, $\{\ell_i\}$, $t_e$.
\end{algorithmic}
\end{algorithm}

By applying the aforementioned rules and using the definition in \eqref{eqn:total_delay}, the suboptimal algorithm for solving Problem (P4) is presented in Algorithm 2 with the key steps described as follows. Specifically, we consider three scenarios of the latency constraint $T$. First,  if $T \in \left[\mathcal{D}(\mathcal{N}_1), +\infty\right)$, the latency constraint in \eqref{eqn:Const} is met by the offloading policy from Steps 1-2, which is thus optimal without the need of further modification. Second, if $T \in \left[\mathcal{D}(\emptyset), \mathcal{D}(\mathcal{N}_1)\right)$, based on the said rules, we remove the members in $\mathcal{S}_1$ incrementally to reduce the total latency by the following iterative procedure. We initialize $\mathcal{S}_1$ as $\mathcal{N}_1$ and continue to remove users from $\mathcal{S}_1$ until $\mathcal{D}\left(\mathcal{S}_1\right)\leq T$ is met. In each removal, we take away one user using greedy strategy, i.e., selecting the user in $\mathcal{S}_1$ with the minimum energy efficiency (i.e., $\frac{-\theta_i}{a_i+b_i\gamma_i}$). Last,  if the given $T$ is still smaller than  $\mathcal{D}(\emptyset)$, we begin to reduce  the offloaded bits of users in $\mathcal{M}_1$ for further  latency reduction. Fortunately, when $\mathcal{S}_1=\emptyset$, Problem (P4) is reduced to the problem of determining $\ell_i$'s in $\mathcal{M}_1$ as follow:
\begin{align}\label{eqn:LP_M1}
\min_{\{\ell_i\}, ~t_e} &\quad \sum_{i\in \mathcal{M}_1} \theta_i \ell_i \\
{\rm{s.t.}} &\quad \sum_{i\in \mathcal{M}_1}\ell_i \left(a_i+ b_i \gamma_i\right)+t_e\leq \widetilde{T}, \nonumber \\
&\quad L_i^\text{min} \leq \ell_i \leq \min \left\{L_i, t_e r_i \left(1+d\right)^{1-|\mathcal{M}|}\right\},\quad \forall i \in \mathcal{M}_1, \nonumber \\
&\quad t_e \geq \max_{i\in \mathcal{M}_0}\left\{\frac{L_i^\text{min}}{r_i(1+d)^{1-|\mathcal{M}|}}\right\}.\nonumber
\end{align}
which is an LP problem and can be solved efficiently by the LP solver.

The complexity of Algorithm 2 is dominated by solving the LP problem in Step 12, which is $\mathcal{O}((|\mathcal{M}_1|)^{3.5})$.

\section{Simulation Results}
In this section, we provide simulation results to evaluate the proposed algorithms. The parameters are set as follows, unless otherwise stated. We set $T=35$ ms and $\omega_i=1$, $\forall i\in \mathcal{K}$. For each user $i$, we set the uplink and downlink transmission rates $a_i^{-1}$ and $b_i^{-1}$ uniformly distributed in $[100, 150]$ Mbps and $[150, 200]$ Mbps, respectively. The computation-service rate $r_i$ follows uniform distribution over $\left[1\times 10^7, 2\times 10^7\right]$ bits/sec. In addition, we set the ratio of output/input data  $\gamma_i=10^{-x}$, where $x$ is uniformly distributed over $\left[0.5, 1.5\right]$.  All random variables are independent for each user and the simulation results are obtained by averaging over 500 realizations.

\subsection{Offloading Rate Maximization}

\begin{figure}[t]
\centering
\begin{minipage}[t]{0.48\textwidth}
\centering
\includegraphics[width=\textwidth]{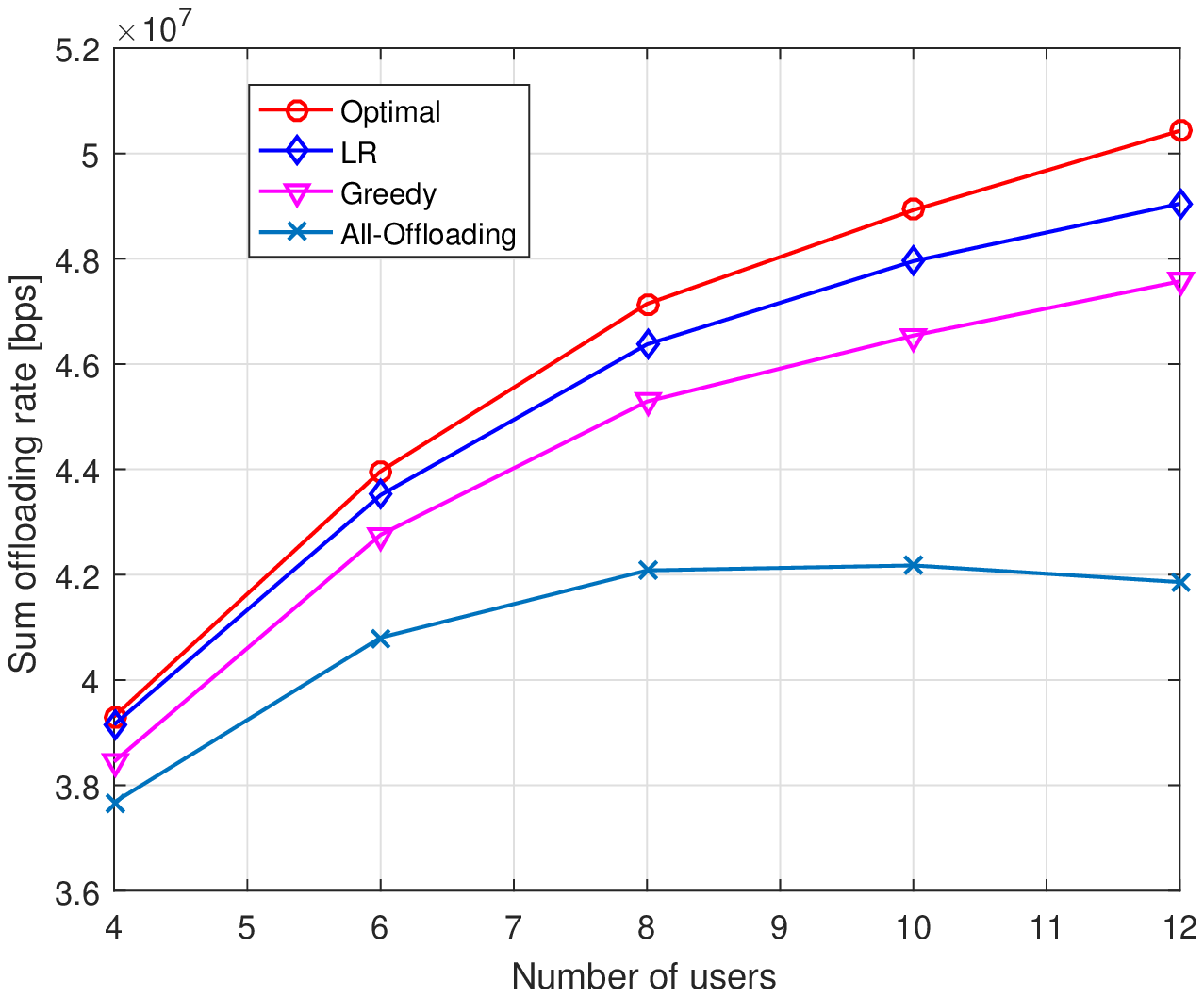}
\vspace*{-10mm}
\caption{Sum offloading rate versus $K$.}\label{fig:f3}
\end{minipage}
\begin{minipage}[t]{0.48\textwidth}
\centering
\includegraphics[width=\textwidth]{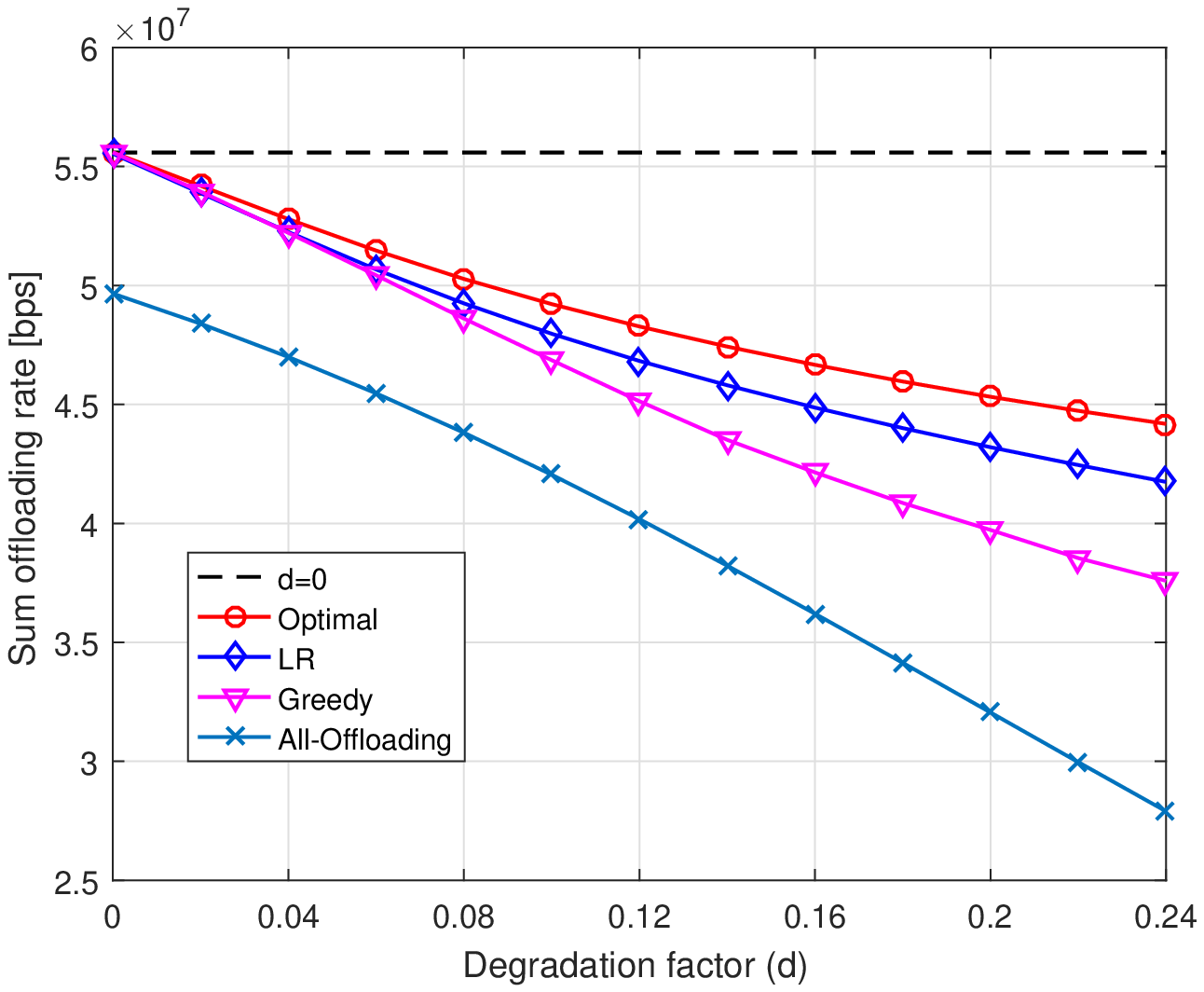}
\vspace*{-10mm}
\caption{Sum offloading rate versus $d$.}\label{fig:f4}
\end{minipage}
\end{figure}

For performance comparison, we introduce three benchmark algorithms in the following.
\begin{enumerate}
\item \emph{All-Offloading:} All the users are scheduled to offload, i.e., $\mathcal{S}=\mathcal{K}$.
\item \emph{Greedy:} $\mathcal{S}$ is obtained through a greedy method, i.e., selecting users in the descending order of the transmission rate (i.e., $\frac{\omega_i}{a_i+b_i\gamma_i}$) until condition \eqref{eqn:S_pro1} is invalid.
\item \emph{Linear Programming Relaxation (LR):} $\mathcal{S}$ is obtained by solving $K$ slave problems in \eqref{eqn:subproblem_of_S} using linear programming relaxation \cite{LPR}.
\end{enumerate}
Note that the three benchmarks are used to find the offloading-user set, then the rest problem is reduced to an LP that can be solved efficiently.

In Fig. \ref{fig:f3}, we compare the sum offloading rate performance of different algorithms when the number of users $K$ varies from $4$ to $12$, where $d$ is set as $0.1$. First, we can see that the sum offloading rate is increasing with $K$ for the optimal, LR and greedy algorithms, while for the scheme that all users offload, it grows slowly when $K\leq 10$ and begins to decrease afterwards. This is because the former three algorithms have more flexible user-scheduling schemes to balance the degradation impact caused by I/O interference and thus have superior system performance. In contrast, the last algorithm with no control on the number of offloading users, will suffer more severe performance degradation as $K$ increases. Besides, it can be observed that the optimal algorithm outperforms the benchmark algorithms especially when $K$ is large. For instance, when $K=12$, the optimal algorithm obtains about $3\%$, $6\%$, and $20\%$ performance improvements over the three benchmarks respectively.

In Fig. \ref{fig:f4}, we illustrate the relationship between the degradation factor $d$ and the sum offloading rate performance, where $K=10$. As expected, the sum offloading rate is decreasing with $d$ in all considered algorithms while the descending rate of the optimal algorithm is the slowest. This indicates that our proposed algorithm has the best performance resistance against the I/O-interference effect. One observes that, the performance of LR and greedy algorithms is close-to-optimal when $d$ is small. This coincides with the result of special case 4) in Section III-C that when the degradation factor $d$ is zero, the optimal solution can be obtained by the greedy approaches. On the other hand, the line of $d=0$ can been seen as the sum offloading rate of the conventional case without considering the I/O interference issue. Its performance gap with the optimal algorithm can be interpreted as the overestimation of the system performance builded on the optimistic assumption of no I/O interference.

\subsection{Energy Minimization}

\begin{figure}[t]
\centering
\begin{minipage}[t]{0.48\textwidth}
\centering
\includegraphics[width=\textwidth]{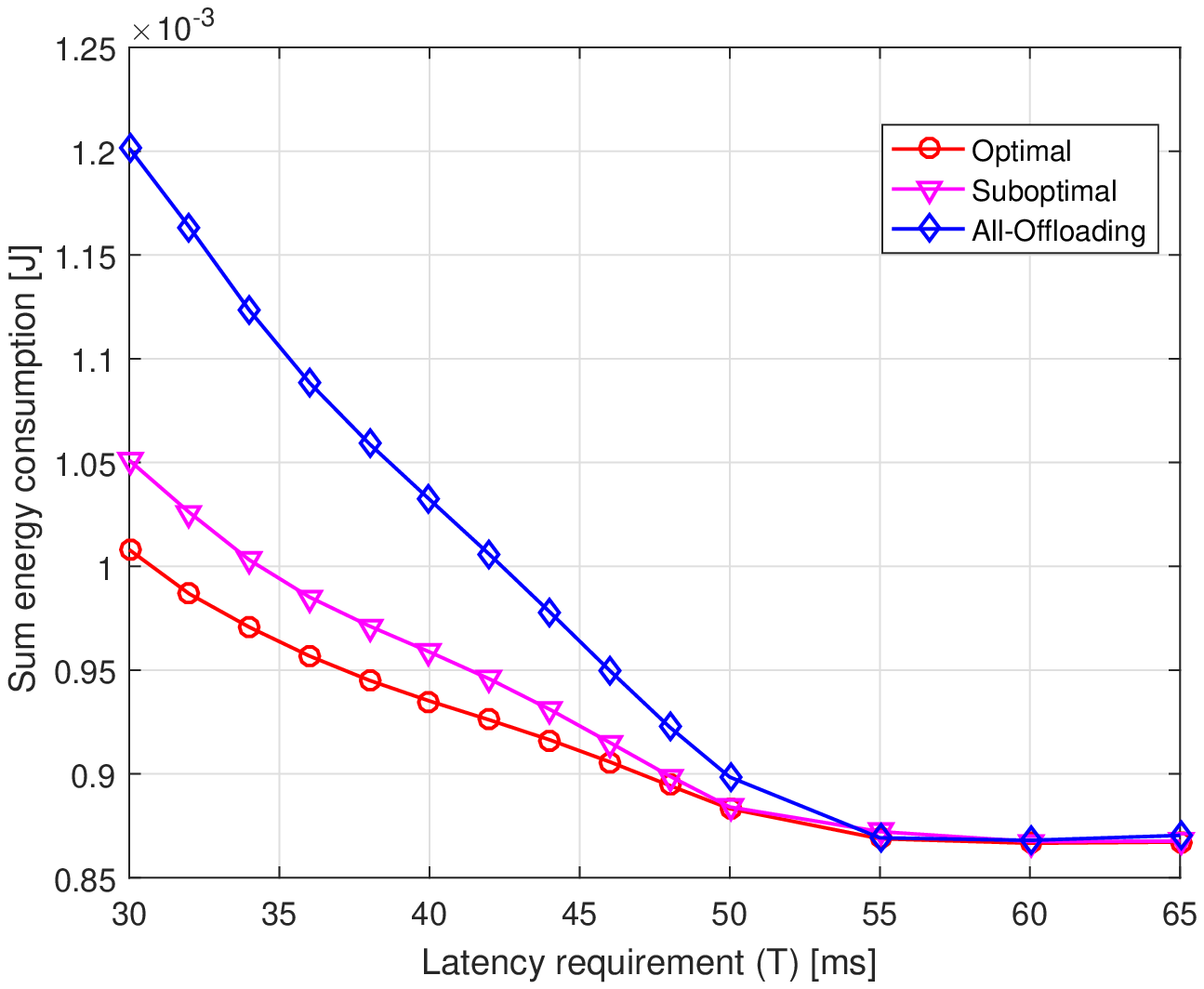}
\vspace*{-10mm}
\caption{Sum energy consumption versus $T$.}\label{fig:f5}
\end{minipage}
\begin{minipage}[t]{0.48\textwidth}
\centering
\includegraphics[width=\textwidth]{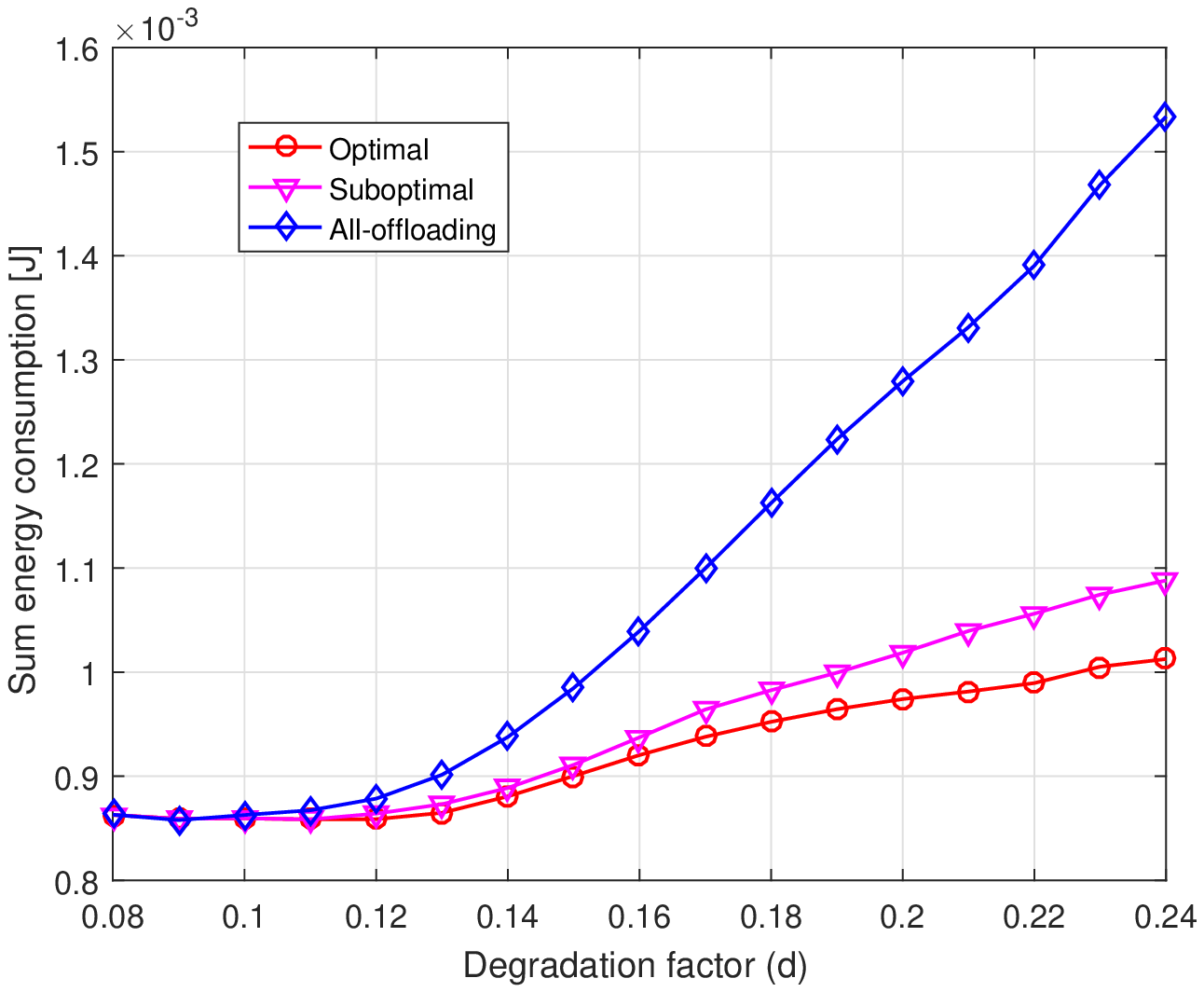}
\vspace*{-10mm}
\caption{Sum energy consumption versus $d$.}\label{fig:f6}
\end{minipage}
\end{figure}

For measuring the energy consumption, we set $\kappa_i=10^{-28}$ \cite{Dynamic_JunZhang} and $p_i=0.1$ W. The data size of task, the required number of CPU cycles per bit, and the local-computing speed follow uniform distribution with $L_i\in [50,100]$ KB, $c_i\in \left[500, 1000\right]$ cycles/bit, and $f_i \in [2\times 10^{8}, 6\times 10^{8}]$ cycles/sec, respectively.

To evaluate the proposed suboptimal algorithm, we present the performance of all-offloading scheme mentioned in the preceding subsection and the optimal performance that is obtained via decomposing Problem (P2) into $K$ mixed-integer linear programming (MILP) subproblems and each subproblem is optimally solved using the MILP solver. The MILP solver implements Branch-and-Bound based algorithms which have complexities exponentially increasing with the size of $K$.

In Fig. \ref{fig:f5}, we investigate the impact of latency constraint $T$ on the sum energy consumption, where $K=10$ and $d=0.2$. First, we see that as the maximum tolerance latency $T$ increases, the sum energy consumption decreases. This is because a more relaxed latency requirement facilitates larger task fraction to offload and consequently saves more energy. However, when $T$ is sufficiently large, the sum energy consumption achieves its minimum and becomes irrespective of $T$ since the optimal offloading scheme always meets the latency constraint in \eqref{eqn:Const}.  Next, it can be observed that the suboptimal algorithm has superior performance compared to the all-offloading scheme especially when $T$ is small. In particular, when $T$ is $30$ ms, the suboptimal algorithm obtains about $14\%$ energy-saving compared to the all-offloading algorithm.

In Fig. \ref{fig:f5}, we present the performance of the proposed suboptimal algorithm versus the degradation factor $d$, where $K=10$.  We  observe that compared with the dramatic increase of sum energy consumption with $d$ in the all-offloading algorithm, the sum energy consumption in the proposed suboptimal algorithm grows at a much slower rate, which is close to the line of the optimal performance.

\section{Conclusions}
In this paper, we investigate joint radio-and-computation resource allocation in a multiuser MEC system, where the computation interference issue has been considered. We formulate two optimization problems: sum offloading rate maximization and sum energy minimization. To address rate maximization, we first solve the optimal offloaded data size and computation time allocation for any given offloading-user set that meets the necessarily optimal condition. Then we develop an optimal algorithm based on Dinkelbach method to find the optimal offloading-user set. For solving energy  minimization, we transform the original problem into an equivalent one that facilitates the scheduling design and propose an algorithm to find a sub-optimal solution. Simulation results demonstrate that our proposed algorithms achieve superior performance gain compared with the benchmark algorithms.


\section*{Appendix}

\appendices
\subsection{Proof of Lemma \ref{lemma1}}
To prove this lemma, it is sufficient to show that for any given $\mathcal{S}$, one of the constraints  $\ell_i \leq t_e r_i \left(1+d\right)^{1-|\mathcal{S}|}$ or $\ell_i\geq 0$ must be active at the optimal $\ell_i^*$, $\forall i\in \mathcal{S}$.

For a given $\mathcal{S}$, since Problem (P1) is an LP problem with a bounded feasible region, there exists an optimal solution located at a vertex (i.e., extreme point)\cite{LP}. Denote $\overline{x}=[\ell_1, \ell_2,...,\ell_{|\mathcal{S}|}, t_e]^T \in \mathbb{R}^{|\mathcal{S}|+1}$ as the vertex that is optimal. By the vertex definition, there are $|\mathcal{S}|+1$ linearly independent active constraints at $\overline{x}$. First, it is easy to check that \eqref{eqn:st_T} should be active at $\overline{x}$. Moreover, the pair of constraints $\ell_i \leq t_e r_i \left(1+d\right)^{1-|\mathcal{S}|}$ and $\ell_i\geq 0$ for each $i\in\mathcal{S}$ should not be active or inactive simultaneously at $\overline{x}$. Both of them can be verified by contradiction as follows.

Suppose the former case is satisfied at $\overline{x}$, then it has $\ell_i=0$ and $t_e=0$, leading to a trivial solution $\overline{x}=\mathbf{0}$ that violates the active condition on \eqref{eqn:st_T}. Next, if the later case is satisfied, i.e., there exists a constraint pair (e.g., $i\in \mathcal{S}$) being inactive simultaneously at $\overline{x}$, we have to select $|\mathcal{S}|$ constraints to be active from the rest $2(|\mathcal{S}|-1)$ constraints in \eqref{eqn:st_ell1}. In this case, another constraint pair (say, $j\in \mathcal{S}$ with $j\neq i$) needs to be active concurrently for composing the active constraint set, which eventually returns back to the former case. Therefore, one and only one of the constraints between $\ell_i \leq t_e r_i \left(1+d\right)^{1-|\mathcal{S}|}$ and $\ell_i\geq 0$ of each $i\in \mathcal{S}$ can be active at $\overline{x}$. This completes the proof.

\subsection{Proof of Proposition \ref{proposition1}}
To prove this, we need the following lemma which can be easily proved.
\begin{lemma}\label{lemma3}
Consider the fractions $x_1/y_1$ and $x_2/y_2$, with $x_i, y_i>0, i=1,2$. Then,
\begin{align*}
\min\left\{\frac{x_1}{y_1}, \frac{x_2}{y_2}\right\} \leq \frac{x_1+x_2}{y_1+y_2} \leq \max\left\{\frac{x_1}{y_1}, \frac{x_2}{y_2}\right\}.
\end{align*}
\end{lemma}

We first prove the proposition from sufficiency. If the given $\mathcal{S}$ satisfies condition \eqref{eqn:S_pro1}, using Lemma \ref{lemma3}, the following inequality holds for all $i\in \mathcal{S}$:
\begin{align}\label{eqn:S_pro1_inequ}
\frac{\omega_i}{(a_i+b_i\gamma_i)}=\frac{\omega_ir_i}{(a_i+b_i\gamma_i)r_i}&\overset{(a)}{\geq} \frac{\sum_{j\in \mathcal{S}\backslash\{i\}}\omega_j r_j +\omega_ir_i}{(1+d)^{|\mathcal{S}|-1}+\sum_{j\in \mathcal{S}\backslash\{i\}}(a_j+b_j\gamma_j)r_j +(a_i+b_i\gamma_i)r_i}\nonumber\\
&\overset{(b)}{\geq}
\frac{\sum_{j\in \mathcal{S}\backslash\{i\}}\omega_jr_j }{(1+d)^{|\mathcal{S}|-1}+\sum_{j\in \mathcal{S}\backslash\{i\}}(a_j+b_j\gamma_j)r_j},
\end{align}
where the right hand side of the first inequality is identical to $R$. The term in the second inequality is the sum offloading rate achieved by setting $\widetilde{\mathcal{S}}=\mathcal{S}\backslash \{i\}$ in \eqref{eqn:tildeS}. (a) holds for all $i\in\mathcal{S}$ since condition \eqref{eqn:S_pro1} is met. (b) is deduced by Lemma \ref{lemma3}. The relation (b) holding for all $i\in \mathcal{S}$ indicates that  $\widetilde{\mathcal{S}}=\mathcal{S}$ has a larger sum offloading rate than any neighbors $\widetilde{\mathcal{S}}=\mathcal{S}\backslash \{i\}$, $\forall i\in\mathcal{S}$. Thus, $\widetilde{\mathcal{S}}=\mathcal{S}$ is the local optimum of Problem \eqref{eqn:tildeS}. According to the results in \cite{Local_Global}, any point of local optimum of Problem \eqref{eqn:tildeS} is also point of global optimum. Therefore, $\widetilde{\mathcal{S}}=\mathcal{S}$ is the optimal solution of Problem \eqref{eqn:tildeS}.

Next, from necessity, since $\widetilde{\mathcal{S}}=\mathcal{S}$ is the optimal solution of Problem \eqref{eqn:tildeS}, the local optimum condition is also met. Thus, we have the relation (b) in \eqref{eqn:S_pro1_inequ} satisfying for all $i\in\mathcal{S}$. Using Lemma \ref{lemma3}, (a) is deduced for any $i\in\mathcal{S}$, i.e., condition \eqref{eqn:S_pro1} holds the given $\mathcal{S}$. This completes the proof.

\subsection{Proof of Proposition \ref{proposition2}}
Let  $\mathcal{S}^*$ and $R^*$ denote the optimal solution and the optimal objective value of the Problem \eqref{eqn:problem_of_S} that has relaxed constraint \eqref{eqn:S_pro1}, respectively. For ease of expression, we sort the entities $\left\{\frac{\omega_i}{a_i+b_i\gamma_i}\right\}$ in $\mathcal{S}^*$ in the descending order $\frac{\omega_1}{a_1+b_1\gamma_1}> \cdots > \frac{\omega_m}{a_m+b_m\gamma_m}$, with $m=|\mathcal{S}^*|$.

The proposition can be proved by contradiction. Suppose that the optimal solution $\mathcal{S}^*$ violates constraint \eqref{eqn:S_pro1}, we have $\frac{\omega_1}{a_1+b_1\gamma_1}>\cdots >\frac{\omega_{j-1}}{a_{j-1}+b_{j-1}\gamma_{j-1}}> R^*> \frac{\omega_j}{a_j+b_j\gamma_j} > \cdots > \frac{\omega_m}{a_m+b_m\gamma_m}$. Denote $\mathcal{S}^\prime= \mathcal{S}^*\backslash \{j\}$. Since $R^*>\frac{\omega_j}{a_j+b_j\gamma_j}$, using Lemma \ref{lemma3}, we have the following  inequality
\begin{align}\label{eqn:Sj_proof3}
R^*&= \frac{\sum_{i\in \mathcal{S}^*}\omega_i r_i}{(1+d)^{|\mathcal{S}^*|-1}+\sum_{i\in \mathcal{S}^*}(a_i+b_i\gamma_i)r_i}\nonumber \\
&\overset{(a)}{<}\frac{\sum_{i\in \mathcal{S}^*\backslash \{j\}}\omega_i r_i}{(1+d)^{|\mathcal{S}^*|-1}+\sum_{i\in \mathcal{S}^*\backslash \{j\}}(a_i+b_i\gamma_i)r_i} \nonumber \\
&\overset{(b)}{<}\frac{\sum_{i\in \mathcal{S}^\prime}\omega_i r_i}{(1+d)^{|\mathcal{S}^\prime|-1}+\sum_{i\in \mathcal{S}^\prime}(a_i+b_i\gamma_i)r_i},
\end{align}
where (a) is deduced from Lemma \ref{lemma3} and (b) is due to $|\mathcal{S}^*|>|\mathcal{S}^\prime|$.  \eqref{eqn:Sj_proof3} shows that sum offloading rate of $\mathcal{S}^\prime$ is larger than that of $\mathcal{S}^*$, which contradicts to the assumption that $\mathcal{S}^*$ is optimal. Notice that $\mathcal{S}^\prime$ may not satisfy \eqref{eqn:S_pro1} at present. However, as long as $\mathcal{S}^\prime$ does not meet condition \eqref{eqn:S_pro1}, we can treat the current $\mathcal{S}^\prime$ as another $\mathcal{S}^*$ and use the same manner in \eqref{eqn:Sj_proof3} to construct a new $\mathcal{S}^\prime$. This guarantees to find $\mathcal{S}^\prime$ that meets \eqref{eqn:S_pro1}, since the extreme case is $\mathcal{S}^\prime$ with $|\mathcal{S}^\prime|=1$ (i.e., single user) that always satisfies \eqref{eqn:S_pro1}.

If $\mathcal{S}$ does not meet \eqref{eqn:S_pro1},  we can always find an $\mathcal{S}^\prime$ satisfying \eqref{eqn:S_pro1} and with larger objective value than $\mathcal{S}$. Thus, it can be concluded that $\mathcal{S}$ violating \eqref{eqn:S_pro1} cannot be the optimal solution of the relaxed problem. This completes the proof.

\subsection{Proof of Lemma \ref{lemma2}}
We provide the sufficiency proof of Lemma \ref{lemma2} since the necessity proof is just reversed. Let $R_m^\prime$ denote the root of $g(R_m)=0$ (note that the existence and uniqueness of the root of $g(R_m)=0$ are proved in \cite{NonLinear}) and $\mathbf{x}^\prime\in \mathcal{F}_m$ an optimal solution of $g(R_m^\prime)$.

Since $g(R_m^\prime)=\max_{x\in \mathcal{F}_m}\left\{N(\mathbf{x})-D(\mathbf{x})R_m^\prime\right\}=0$, we have
\begin{align}\label{eqn:proof_lemma2}
N(\mathbf{x})-D(\mathbf{x})R_m^\prime \leq N(\mathbf{x}^\prime)-D(\mathbf{x}^\prime)R_m^\prime=0, \quad \forall x\in \mathcal{F}_m.
\end{align}
As $D(\mathbf{x}), D(\mathbf{x}^\prime)>0$, $\forall\mathbf{x}\in \mathcal{F}_m$, \eqref{eqn:proof_lemma2} can be re-written as   $R_m^\prime=\frac{N(\mathbf{x}^\prime)}{D(\mathbf{x}^\prime)}\geq \frac{N(\mathbf{x})}{D(\mathbf{x})}, \forall \mathbf{x}\in \mathcal{F}_m$, i.e., $R_m^\prime=R_m^*$ and $\mathbf{x}^\prime=\mathbf{x}^*$ are the optimal objective value and an optimal solution of Problem \eqref{eqn:subproblem_of_S}, respectively. This completes the proof.

\subsection{Proof of Proposition \ref{proposition4}}
With $\frac{1}{a_1+b_1\gamma_1}=... = \frac{1}{a_K+b_K\gamma_K}=\frac{1}{a+b\gamma}$, th sum offloading rate $R$ can be simplified as
\begin{align*}
R=\frac{\sum_{i\in \mathcal{S}}r_i}{(1+d)^{|\mathcal{S}|-1}+(a+b\gamma)\sum_{i\in\mathcal{S}}r_i}=\left(\frac{(1+d)^{|\mathcal{S}|-1}}{\sum_{i\in \mathcal{S}}r_i}+(a+b\gamma)\right)^{-1}.
\end{align*}
Therefore, for maximizing $R$, it is sufficient to minimize $\frac{(1+d)^{|\mathcal{S}|-1}}{\sum_{i\in \mathcal{S}}r_i}$ instead. Observe that for a given $|\mathcal{S}|$, the minimum  $\frac{(1+d)^{|\mathcal{S}|-1}}{\sum_{i\in \mathcal{S}}r_i}$ is achieved by selecting $|\mathcal{S}|$ largest $r_i$'s. For simplicity, we notate $|\mathcal{S}|=n$ and sort $r_i$'s in the descending order $r_1\geq \cdots \geq r_K$. Then the problem becomes finding the minimum point in the sequence $\{f_{n}\triangleq\frac{(1+d)^{n-1}}{\sum_{i=i}^{n}r_i}\}$, with $n=1,\cdots, K$.  It can be checked that the sequence $\{f_{n}\}$ has a monotone property with $n$. Specifically, there exists an index $n_0$, in which $\{f_{n}\}$ is monotonically decreasing when $1\leq n\leq n_0$ and monotonically increasing when $n_0\leq n\leq K$. Therefore, $f_{n_0}$ is the minimum point of sequence $\{f_{n}\}$. By defining $r_0=0$ and letting $f_{n}\leq f_{n-1}$, we derive condition $r_{n}\geq d\sum_{i=0}^{{n}-1}r_i$ and $n_0$ is the largest index satisfying this condition. Thus, $\mathcal{S}=\left\{i~| 1\leq i\leq n_0\right\}$ is the optimal offloading-user set. This completes the proof.

\subsection{Proof of Proposition \ref{proposition5}}
For solving the minimum $T$ of Problem (P3), it can be observed from \eqref{eqn:st_T2} that $t_e$ should be minimized. By \eqref{eqn:st_S2}, the minimum $t_e$ can be expressed as $\max_{i\in \mathcal{S}}\left\{\frac{\ell_i}{r_i (1+d)^{(1-|\mathcal{S}|)}}\right\}$. Meanwhile, since $\mathcal{S}=\{i|\ell_i>0, i\in \mathcal{K}\}$, Problem (P3) can be refined as
\begin{align}
\min_{ {\{\ell_i\}}, ~T} &\quad~ T   \label{eqn:ECM_prof4_probl} \\
{\rm{s.t.}}&\quad \sum_{i\in \mathcal{K}} \left(a_i+ b_i \gamma_i\right)\ell_i + \max_{i\in \mathcal{K}}\left\{\frac{\ell_i}{r_i (1+d)^{1-\sum_{i\in \mathcal{K}}\mathds{1}\{\ell_i>0\}}}\right\}\leq T, \label{eqn:ECM_prof4_Toffl} \\
&\quad \frac{c_i(L_i-\ell_i)}{f_i}\leq T, \quad \forall i\in \mathcal{K}, \label{eqn:ECM_prof4_Tloc}\\
&\quad 0\leq \ell_i \leq L_i , \quad \forall i\in \mathcal{K}, \label{eqn:ECM_prof4_ell}
\end{align}
where $|\mathcal{S}|\triangleq \sum_{i\in \mathcal{K}}\mathds{1}\{\ell_i>0\}$ and $\mathds{1}\{\cdot\}$ is a binary indicator function.

Combining \eqref{eqn:ECM_prof4_Tloc} and \eqref{eqn:ECM_prof4_ell}, the minimum offloaded bits is obtained as $L_i^{\text{min}} = \left[L_i-\frac{Tf_i}{c_i}\right]^+$, $\forall i$. Notice that here $L_i^{\text{min}}$ is a function of $T$. For notation simplicity, we notate $\sum_{i\in \mathcal{K}}\mathds{1}\{\ell_i>0\}= \sum_{i\in \mathcal{K}}\mathds{1}\{L_i^{\text{min}}>0\}\triangleq N(T)$. Then Problem \eqref{eqn:ECM_prof4_probl} can be simplified as
\begin{align}
\min_{ {\{\ell_i\}}, ~T} &\quad~ T   \label{eqn:ECM_prof4_prob2} \\
{\rm{s.t.}}&\quad \sum_{i\in \mathcal{K}}\left(a_i+ b_i \gamma_i\right)L_i^{\text{min}}+\max_{i\in \mathcal{K}}\left\{\frac{L_i^{\text{min}}}{r_i(1+d)^{1-N(T)}}\right\}\leq T, \label{eqn:ECM_feasi3}
\end{align}
where the left hand side of \eqref{eqn:ECM_feasi3} is the minimum time required for completing the minimum offloaded tasks for a given $T$. Since it is monotonically decreasing with $T$ in the interval of $T\in \left[0, \max_{i\in\mathcal{K}}\left\{c_iL_i/f_i\right\}\right]$, it is easily observed that the minimum of Problem \eqref{eqn:ECM_prof4_prob2} is achieved only when \eqref{eqn:ECM_feasi3} holds with equality. This completes the proof.

\bibliographystyle{IEEEtran}
\bibliography{link}

\end{document}